\documentclass{llncs}
\usepackage{etex}

\usepackage{makeidx}  

\usepackage{verbatim}
\usepackage{graphicx}
\usepackage{epstopdf}
\usepackage{bussproofs}
\usepackage{color} 
\usepackage{dsfont}
\usepackage{stmaryrd}
\usepackage{mathtools}
\usepackage{amssymb}
\usepackage{hyperref}
\usepackage{enumerate}
\usepackage[all]{xy}
\usepackage{relsize}
\usepackage{longtable}
\usepackage[]{units}
\usepackage{tikz}
\usetikzlibrary{matrix}
\usepackage{subfig}
\usepackage{array}
\usetikzlibrary{calc,arrows,positioning}
\usepackage{color}  
\usepackage{framed}  

\newtheorem{teorema}{Teorema}[section]

\newtheorem{convention}[teorema]{Convention}

\definecolor{shadecolor}{gray}{0.9}

{
\begin{figure}[h!]
\fbox{%
           \begin{minipage}{0.460\textwidth}
#1
\end{minipage}}
\end{figure}
 }

\newcommand{\bydef}{ \stackrel{\mathrm{def}}{=} }

\newcommand{\freccia}[1]{\stackrel{#1}{\longrightarrow}}

\newcommand{\lfp}{\textnormal{lfp}}
\newcommand{\gfp}{\textnormal{gfp}}

\newcommand{\R}{\textbf{R}}
\newcommand{\K}{\textbf{K}}

\newcommand{\sem}[1] {  \llbracket #1 \rrbracket  }  
\newcommand{\gsem}[1]{\llparenthesis{\,#1\,}\rrparenthesis}

\newcommand\node[1]{*+[o]{#1}}

\newcommand\addLabelUL[1]{\ar@{}[]+UR|(1){~\makebox[0pt][l]{$\mathbf{#1}$}}}

\newcommand\addLabelUR[1]{\ar@{}[]+UR|(1){~\makebox[0pt][l]{$\mathbf{#1}$}}}

\newcommand\addLabelDL[1]{\ar@{}[]+UR|(1){~\makebox[0pt][l]{$\mathbf{#1}$}}}

\newcommand\addLabelDR[1]{\ar@{}[]+UR|(1){~\makebox[0pt][l]{$\mathbf{#1}$}}}

\newcommand\addDMD[2]{
	\ar@{-}[]+<#1pt,0pt>;[]+<0pt,#2pt>
	\ar@{-}[]+<0pt,#2pt>;[]+<-#1pt,0pt>
	\ar@{-}[]+<-#1pt,0pt>;[]+<0pt,-#2pt>
	\ar@{-}[]+<0pt,-#2pt>;[]+<#1pt,0pt>
}

\newcommand{\catset}{{\normalfont\textbf{Set}}}

\begin{document}
\frontmatter          
\mainmatter              

\title{Upper-Expectation Bisimilarity\\ and Real-valued Modal Logics}
%
%
\author{Matteo Mio \thanks{
The author would like to thank Alex Simpson, Jan Rutten, Marcello Bonsangue,  Helle Hvid Hansen, Henning Basold, Alexandra Silva and the anonymous reviewers for their helpful comments and suggestions.}
\thanks{
The author carried out this work during the tenure of an ERCIM ÒAlain BensoussanÓ Fellowship, supported by the Marie Curie Co-funding of Regional, National and International Programmes
 (COFUND) of the European Commission.}
}

\authorrunning{Matteo Mio} 
%
%
\institute{CWI-Amsterdam,\\
\email{miomatteo@gmail.com}}

\maketitle              

\begin{abstract}
Several notions of bisimulation relations for probabilistic nondeterministic transition systems 
have been considered in the literature. We consider a novel testing-based behavioral equivalence called \emph{upper-expectation bisimilarity} and develop its theory using standard results from linear algebra and functional analysis. We show that, for a wide class of systems, our new notion coincides with Segala's \emph{convex bisimilarity}. We develop logical characterizations in terms of expressive probabilistic modal $\mu$-calculi and a novel real-valued modal logic. We prove that upper-expectation bisimilarity is a congruence for the wide family of process algebras specified following the probabilistic GSOS rule format.

\keywords{Quantitative Modal Logics, Bisimulation, Convexity.}
\end{abstract}

\section{Introduction}\label{sec_introduction}
Directed-graph structures are sufficient for modeling nondeterministic programs and concurrent systems but can not be used to represent other important aspects of computation, such as \emph{probabilistic behavior}, \emph{timed transitions} and other \emph{quantitative} information one might need to express. To address this limitation, since the late 80's, a lot of research has focused on the identification of appropriate structures for expressing these quantitative aspects  (see, e.g., \cite{LS91,H94,BA1995}), and in particular for modeling probabilistic behavior. One of the most successful such models  is today known under several names: (simple) Segala systems \cite{BartelsThesis}, concurrent Markov chains \cite{H94}, probabilistic automata \cite{S95} or just \emph{probabilistic nondeterministic transition systems} (PNTS). Today PNTS's are the mathematical structures, generalizing standard nondeterministic transition systems (NTS), most often used to provide  operational semantics to probabilistic and nondeterministic languages \cite{HP2000,KNPV2009,BartelsThesis}. 

A central concept in the theory of programming languages and concurrent systems is the notion of \emph{behavioral equivalence}. An equivalence relation $\simeq\ \subseteq\! S\times S$ between states of a system is, informally speaking, a behavioral equivalence if $s\simeq t$ implies that $s$ and $t$ satisfy the same class of properties of interest. Of course different classes of properties induce different notions of equivalence. The paradigmatic example of behavioral equivalence for ordinary NTS's is Milner and Park's \emph{bisimilarity} \cite{M80}. Among its good properties, bisimilarity enjoys the following: $\textbf{B1)}$ two states are bisimilar if and only if they satisfy the same properties expressed, e.g., in the modal $\mu$-calculus \cite{Kozen83} or in other (weaker but useful in practice) branching-time logics such as CTL, $\textnormal{CTL}^{*}$ \cite{Stirling96}, and the basic modal logic $\K$ \cite{BdRVModal} or its labeled version, the Hennessy--Milner modal logic \cite{HM85}. %
Furthermore, $\textbf{B2)}$ bisimilarity is a congruence for a wide family of process algebras, all specified following one of the many rule formats: GSOS \cite{BIM95}, tyft/tyxt \cite{GV1989}, \emph{etc}. This means that if programs $p$ and $q$ are bisimilar then, for any program context $C[x]$, the two composite programs $C[p]$ and $C[q]$ are also bisimilar. This is the crucial property that authorizes the substitution of equals for equals in programs, a process of fundamental importance in, e.g., optimization and compositional program development and verification. $\textbf{B3)}$ Bisimilarity also enjoys a rich mathematical theory based on coinduction and plays a role of paramount importance in \emph{coalgebra} \cite{SanRut2011}. 
Lastly, but importantly, $\textbf{B4)}$ bisimilarity can also be explained in terms of Milner's standard metaphor of \emph{push-buttons experiments} on systems \cite{M80}. Such experiments provide an abstract, yet intuitive, testing semantics for bisimilarity.

In the context of PNTS's several notions of behavioral equivalence, based on the technical machinery of coinduction, have been considered in the literature \cite{S95,LS95,SZG2011,BdNM2012,deAlfaro2008,DvG2010}. 
The one which has attracted most attention so far was introduced by Segala in \cite{S95} and is referred to in this paper as \emph{standard bisimilarity}. This is a mathematically natural notion (cf. B3) and, indeed, it has been rediscovered using the methods of coalgebra theory (see., e.g., \cite{Sokalova2011}). Importantly, standard bisimilarity is a congruence relation (cf. B2) for the wide class of  PGSOS process algebras, which virtually includes all Milner's CCS-style (probabilistic) process operators of practical interest \cite{BartelsThesis}. However standard bisimilarity happens to be strictly finer than the equivalence induced by important temporal logics for expressing useful properties of PNTS, such as,  PCTL, $\textnormal{PCTL}^*$ \cite{BA1995,BaierKatoenBook} and the fixed-point modal logics ($\mu$-calculi)  of \cite{HM96,deAlfaro2008,MM07,MioThesis}. This is a rather unsatisfactory fact (cf. B1) because one does not want to distinguish between programs that satisfy the same properties of interest. Furthermore no (widely accepted) testing semantics (cf. B4) for standard bisimilarity exists.

In \cite{S95} another behavioral equivalence for PNTS's, which we refer to as \emph{convex bisimilarity}, is introduced by Segala. This notion, as we shall  discuss in Section \ref{convex_bisimilarity_section}, is strongly motivated by the concept of \emph{probabilistic scheduler}. Convex bisimilarity is coarser than standard bisimilarity. However it is known to be strictly finer than the equivalence induced by the logic $\textnormal{PCTL}^*$ which in turn admits a coinductive characterization \cite{SZG2011}. Other notions of bisimulation for PNTS's have been recently introduced in \cite{DvGHM2009},  \cite{CR2011} and \cite{BdNM2012}.


\paragraph{Contributions.}
In this paper we analyze yet another notion of behavioral equivalence which we refer to as \emph{upper expectation ($\textnormal{UE}$) bisimilarity}. By application of known results in linear algebra and functional analysis, we show that UE-bisimilarity and convex bisimilarity coincide under very mild assumptions. As it is often the case, having two complementary descriptions of the same concept turns out to be useful. As we discuss in Section \ref{section_UE_bisim_intro}, UE-bisimilarity arises naturally from a very abstract testing scenario based on $\mathbb{R}$-valued experiments (cf. property \textbf{B4} above). Unlike other similar works, our experiments are not given by, e.g., the formulas of a given logic (such as the $\mu$-calculus of \cite{deAlfaro2008}) nor by terms of some process algebra (see, e.g., \cite{DvGHM2009}) but, instead, are modeled by  functions $f\!:\!S\!\rightarrow\!\mathbb{R}$ from program states to real values. 

We argue in Section \ref{real_valued_section} that this abstraction, besides being mathematically convenient, sheds light on the mathematical foundations of real-valued logics (including the $\mu$-calculi mentioned above) which have been subject of increasing interest in the last decade \cite{MM07,PrakashBook,FGK2010,AM04,MIO2012b}. In Section \ref{logic_R} we develop the theory of a novel real-valued modal logic, which we call $\R$ (in honor of Frigyes Riesz), inspired by the powerful results from functional analysis presented in Section \ref{sec_representation_theorems}. We define a model-theoretic semantics for $\R$ based on PNTS's and an algebraic semantics. We prove that the two coincide. To the best of our knowledge, no algebraic (i.e., axiomatic) semantics of other $\mathbb{R}$-valued modal logics for PNTS's have appeared in the literature. 
In standard modal logic, particular classes of models (e.g., reflexive NTS's) can be obtained by adding axioms (e.g.,  $\Box \phi\rightarrow \phi$) to the base modal logic $\K$ \cite{BdRVModal}.  We show that Markov processes and ordinary NTS's, which can be seen as particular classes of PNTS's, can be captured by adding appropriate axioms to the basic logic $\R$. This confirms the naturalness of $\R$ which we see as one of the main contributions of this work.

The logic $\R$ is valuable for its mathematical simplicity but cannot express useful properties of systems such as: termination goals, liveness constraints, \emph{etc}. The $\mathbb{R}$-valued modal  $\mu$-calculi of \cite{MioThesis} are better suited to this task, as they (strictly) subsume the probabilistic logic PCTL of \cite{BA1995} which is the specification language adopted in, e.g., the PRISM verification framework \cite{PRISM4}. In Section \ref{logic_mu_calculi} we prove that UE-bisimilarity coincides with the equivalence relation induced by the $\mu$-calculi of \cite{MioThesis}. Thus the quantitative approach to probabilistic $\mu$-calculi may be considered equally suitable as a mechanism for  characterising process equivalence as other non-quantitative $\mu$-calculi
advocated for this purpose (see, e.g.,  \cite{DvG2010}). We also provide in Section \ref{metrics_section} a logical characterization of the Hausdorff behavioral metric for PNTS's, a concept of fundamental importance in the theory of approximation of probabilistic systems pioneered by Panangaden \cite{PrakashBook}.
Collectively, the results of Section \ref{real_valued_section} provide strong logical foundations for UE-bisimilarity (cf.\, property \textbf{B1})  comparable with those of Milner and Park's bisimilarity.

We prove in Section \ref{congruence_section} that UE-bisimilarity is a congruence relation with respect to the CCS-style (communicating) parallel composition operator. This result is known from \cite{LS95} where the property is proved for (the equivalent under mild conditions) convex bisimilarity and parallel composition is defined in terms of an equivalent automata-theoretic synchronized product operation. Our proof, however, is interesting because it uses process-algebra methods and readily extends to every process-operator definable in the PGSOS format of \cite{BartelsThesis} which virtually includes all CCS-style process operators of practical interest (cf.\, property \textbf{B2}).  

In Sections \ref{section_coalgebra} and \ref{section_bisimulations}, summarizing standard and well known results, we discuss how UE-bisimilarity can be understood coalgebraically in terms of cocongruence for an appropriate functor. This confirms the naturalness and mathematical robustness (cf.\, property \textbf{B3}) of this notion of behavioral equivalence.

It seems likely that our results will be of help towards the difficult task of designing and verifying probabilistic concurrent programs and protocols.  The importance of bisimulation in software development (e.g., refinement techniques) and verification (e.g., system minimization) is well known. The logical characterizations of bisimulation are of key importance (see, e.g., \cite{Simpson04,DG2000})  in the development of \emph{compositional} verification methods, which are crucial in the analysis of industrial-size software. Congruence results, such as the one we obtained for PGSOS process algebras (which can be thought of as abstract prototypes of concurrent programming languages) constitute the basis for the development of design and verification techniques based on equational reasoning.  We also suggest that our natural testing semantics for UE-bisimilarity might be related to information-theoretic notions of attacks (cf.\ Section \ref{metrics_section})  which have received much attention in the recent literature  \cite{KPP2008}.

\section{Elementary Notions of Coalgebra}\label{section_coalgebra}
We start by recalling the definition of Milner and Park's bisimulation in the context of ordinary nondeterministic transition systems.

\begin{definition}[\cite{M80}]
A \emph{nondeterministic transition system} (NTS) is a pair $(X,\rightarrow)$ where $X$ is a set of states and $\rightarrow\ \subseteq X\times X$ is a transition relation. An equivalence relation $E\subseteq X\times X$ is a \emph{bisimulation} if $(x,y)\!\in\! E$ implies that:
\begin{itemize}
\item if $x\rightarrow x^\prime$ then there exists $y\rightarrow y^\prime$ such that $(x^\prime,y^\prime)\!\in\! E$, and
\item if $y\rightarrow y^\prime$ then there exists $x\rightarrow x^\prime$ such that $(x^\prime,y^\prime)\!\in\! E$.
\end{itemize}
States $x$ and $y$ are \emph{bisimilar} if $(x,y)\!\in\!E $ for some bisimulation $E$.
\end{definition}
In applications one most often encounters  \emph{labeled} NTS's, where an $L$-indexed set $\{\freccia{a}\}_{a\in L}$ of relations is considered for some set of labels $L$, or \emph{Kripke structures}, where the NTS is endowed with a set of propositional letters  interpreted as predicates. 
For the sake of simplicity we just consider plain NTS's and their probabilistic generalizations. The results we develop extend straightforwardly to the labeled (see, e.g., Section \ref{congruence_section}) and propositional extensions. 

It is going to be convenient to employ the basic language of coalgebra \cite{SanRut2011} in the description of systems and behavioral equivalences. 
We refer to \cite{SanRut2011} for a gentle introduction to the subject. 

\begin{definition}\label{powerset_functor}
Let $\catset$ be the category of sets and functions between them. The endofuntor $\mathcal{P}$ (powerset) on $\catset$ is defined as:
\begin{itemize}
\item $\mathcal{P}(X)=\big\{ A \ | \ A\subseteq X\big\}$, 
\item $\big(\mathcal{P}(f)\big)(A) = f[A] = \{ f(x)\ | \ x\in A\}$ 
\end{itemize}
for all sets $X,Y$ and functions $f:X\rightarrow Y$. 
\end{definition}
We can now restate the definition of NTS's using coalgebraic terminology. The equivalence between the two definitions is well known (see, e.g., \cite{SanRut2011}) and straightforward to verify.
\begin{definition}
A \emph{nondeterministic transition system} (NTS) is a $\mathcal{P}$-coalgebra $(X,\alpha:X\rightarrow\mathcal{P}(X))$.
\end{definition}

\begin{definition}
Let $F,G$ be two endofunctors on $\catset$. A \emph{natural transformation} from $F$ to $G$, written $\eta\!:\!F\!\rightarrow \!G$, is a collection of functions $\eta_{X}\!:\! F(X)\!\rightarrow\! G(X)$, indexed by sets $X\!\in\! \catset$, such that for all functions $f\!:\!Y\!\rightarrow\! Z$ it holds that $\eta_Z\circ G(f)\! =\! F(f)\circ \eta_{Y}$. 
\end{definition}

Thus, a natural transformation is a method $\eta$ for transforming the action of a functor $F$ into that of the functor $G$ in a uniform way.  Following standard ideas (see, e.g., \cite{Sokalova2011}), we shall use this concept (cf. Propositions \ref{nat_transf_3} and \ref{NTS_as_PNTS}) to formalize the idea of transformation of systems of some type into systems of another type. 
Coalgebra provides an abstract notion of behavioral equivalence based on the categorical concept of cocongruence \cite{KurzPHD,DDLP2006}.

\begin{definition}\label{def_quotient_map}
Let $X\in \catset$ and $E\subseteq X\times X$ be an equivalence relation. The map $q_E:X\rightarrow X/E$, mapping elements $x\in X$ to their $E$-equivalence class $[x]_{E}$, is called the \emph{quotient map} of $E$.
\end{definition}

\begin{definition}\label{cocongruence}
Given a \catset-endofunctor $F$ and a $F$-coalgebra $(X,\alpha)$, an equivalence relation $E\!\subseteq\! X\times X$ is a called a \emph{cocongruence} if
for all $(x,y)\!\in\! E$ it holds that $\big(\alpha(x), \alpha(y)\big)\!\in\! \hat{E}$, where the \emph{lifted relation}  $\hat{E}\!\subseteq\! F(X)\times F(X)$ is defined as $\textnormal{ker}(F(q_E))=\{ (A,B) \ | \ F(q_E)(A) = F(q_E)(B) \}$.
\end{definition}

When a functor $F$ is fixed, the definition can be made explicit by expanding the definition of $\hat{E}$. Consider, e.g., the functor $\mathcal{P}$, an equivalence relation $E\subseteq X\times X$ and two sets $A,B\in\mathcal{P}(X)$ with $A=\{x_i\}_{i\in I}$ and  $B=\{y_j\}_{j\in J}$. Then $(A,B)\in \hat{E}$ holds iff the two sets $A/E=\{Ê[x_i]_E\}_{i\in I}$ and $B/E=\{Ê[y_j] \}_{j\in J}$ are \emph{equal} as sets. Thus the relation $E$ is a cocongruence if, for every $(x,y)\in E$, the two sets $\alpha(x)$ and $\alpha(y)$ of reachable states are \emph{equal modulo $E$}. Hence, cocogruences are those equivalence relations that preserve the transition (i.e., coalgebraic) structure of the system. Note that $(A,B)\!\in\! \hat{E}$ holds iff for all $x\!\in\! A$ there exists $y\!\in\! B$ such that $(x,y)\!\in\! E$ and \emph{viceversa}. Therefore the abstract notion of cocongruence  for the functor $\mathcal{P}$ coincides with ordinary Milner and Park's bisimulation for NTS's. 


\section{PNTS's and Bisimulations}\label{section_bisimulations}

\begin{definition}[\cite{Sokalova2011}]\label{functor_distribution}
The endofuntor $\mathcal{D}$ (discrete probability distributions) on $\catset$ is defined as:
\begin{itemize}
\item $\mathcal{D}(X)= \{Ê\mu\!:\!X\rightarrow [0,1] \ | \sum_{x} \mu(x)=1     \}$
\item $\big(\mathcal{D}(f)\big)(\mu)= f[\mu]$,  $ \ \ y \stackrel{f[\mu]}{\longmapsto} \sum \{Ê\mu(x) \ | \  x\in  f^{-1}(y)\} $
\end{itemize}
for all sets $X,Y$ and functions $f\!:\!X\!\rightarrow\! Y$. For $x\!\in\! X$ we denote with $\delta(x)$ (or $\delta_x$) the probability \emph{Dirac} distribution specified by $\delta_x(y)=1$ if $y=x$ and $0$ otherwise. For a set $A\subseteq X$, we write $\mu(A)$ for $\sum_{x\in A}\mu(x)$.
\end{definition}

Note that the composite functor $\mathcal{P}\mathcal{D}$ maps a set $X$ to the collection of all sets of discrete probability distributions on $X$.
\begin{definition}
A \emph{probabilistic nondeterministic transition system} (PNTS) is a $\mathcal{P}\mathcal{D}$-coalgebra $(X,\alpha\!:\!X\!\rightarrow \!\mathcal{P}\mathcal{D}(X))$. We write $x\rightarrow \mu$ to specify that $\mu\in \alpha(x)$.
\end{definition}
The intended interpretation is that the system, at some state $x\!\in\! X$, can evolve by nondeterministically choosing one of the \emph{accessible} probability distributions $\mu$, i.e., such that $\mu\!\in\! \alpha(x)$, and then continuing its execution from the state $y\!\in\! X$ with probability $\mu(y)$. PNTS's can be visualized, using  graphs labeled with probabilities. For example the PNTS $(X,\alpha)$ having set of states $X\!=\!\{x,y\}$ and transition map $\alpha(x)\!=\!\{ \mu_1,\mu_2\}$ and $\alpha(y)=\emptyset$, with $\mu_{1}(x)\!=\!\mu_{1}(y)\!=\!\frac{1}{2}$ and $\mu_{2}(y)\!=\!1$, can be depicted as:
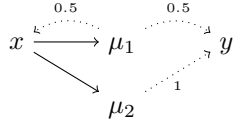
\begin{center}
\begin{tikzpicture}[scale=0.3]
\matrix(a)[matrix of math nodes,
row sep=1em, column sep=2.5em,
text height=1.5ex, text depth=0.25ex]
{x & \mu_1 &y\\
&\mu_2 \\ };
\path[->](a-1-1) edge (a-1-2);
\path[->](a-1-1) edge (a-2-2);
\path[dotted,->,bend right](a-1-2) edge node[above]{ {{\tiny $0.5$}}} (a-1-1);
\path[dotted,->, bend left](a-1-2) edge node[above]{ {{\tiny $0.5$}}} (a-1-3);
\path[dotted,->](a-2-2) edge  node[below]{ {{\tiny $1$}}} (a-1-3);
\end{tikzpicture}
\end{center}
\vspace{-8pt}

The combination of nondeterministic choices immediately followed by probabilistic ones, allows the modeling of  concurrent probabilistic programming languages in a natural way \cite{BartelsThesis,HP2000}. 
As remarked earlier, one often considers labeled or propositional generalizations of PNTS's. All our results extend to these (and other similar) generalizations in a straightforward way (cf. Section \ref{congruence_section}).

PNTS's with a definition equivalent to the coalgebraic one given above were introduced by Segala \cite{S95} who also defined two notions of bisimilarity for PNTS's. One, the stronger (i.e., finer) of the two, which we refer to as  \emph{standard bisimilarity} here, is a natural extension of Milner and Park's bisimilarity for NTS's based on the insights of Larsen and Skou \cite{LS91}.

\begin{definition}[Standard Bisimulation]\label{def_standard_bisim}
Given a PNTS $(X,\alpha)$, a \emph{standard bisimulation} is an equivalence relation $E\subseteq X\times X$ such that if $(x,y)\in E$ then 
\begin{itemize}\item if $x\rightarrow \mu$ then there exists $\nu$ such that $y\rightarrow \nu $ and $\mu=_{E}\nu$, and
\item if $y\rightarrow \nu$ then there exists $\mu$ such that $x\rightarrow \mu $ and $\mu=_{E}\nu$,
 \end{itemize}
 where $\mu=_{E}\nu$ holds if $\mu(A)=\nu(A)$ (see Definition \ref{functor_distribution}) for all sets $A\subseteq X$ which are unions of $E$-equivalence classes.
\end{definition}
\begin{remark}
The notion of standard bisimilarity coincides with that of cocongruence for the functor $\mathcal{PD}$. This is simple to verify by expanding the definition of $\hat{E}$ as in Definition \ref{cocongruence}. For two sets $A=\{\mu_i\}_{i\in I}$ and $B=\{\nu_j\}_{j\in J}$ we have $(A,B)\in \hat{E}$ if and only if $A/E\!=\!\{ [\mu_i]_E\}_{i\in I}$ and  $B/E\!=\! \{ [\nu_j]_E\}_{j\in J}$ are equal as sets where, for each $\mu\in\mathcal{D}(X)$, we denoted with $[\mu]_E$ the probability distribution $q_{E}[\mu]$ (see definitions  \ref{def_quotient_map} and \ref{functor_distribution}).
\end{remark}
The definition looks technical but has a simple interpretation. Two states $x,y$ of a PNTS $( X,\alpha)$ are standard-bisimilar if the two sets of reachable distributions $\alpha(x)$ and $\alpha(y)$ are equal modulo bisimulation, i.e., are equal once their elements ${\mu,\nu}\!\in\!{\mathcal{D}(X)}$ are considered modulo $E$. 
At the level of probability distributions, $\mu\!=_E\! \nu$ means that if one glues together (i.e., identifies) $E$-related states (i.e., applies the quotient map $x\mapsto[x]_{E}$), then $\mu$ and $\nu$ become \emph{equal}, i.e., they assign the same probabilities to all events $A \subseteq X/E$.  As an example, consider the two PNTS's rooted at $x$ and $y$ respectively, depicted in  Figure \ref{example_pic_1}, and assume the processes $x_1$ and $x_2$ to be observationally different\footnote{\label{obs_diff}By this we mean that $(x_1,x_2)\not\in E$ for all bisimulations $E$. Of course this can be implemented by adding structure to the system (e.g., a single edge from $x_1$ to $\mu_1$). Our assumption thus simply abstracts away from such additional details.}. The processes $x$ and $y$ are not standard bisimilar. This is because $y$ can lead to a probability distribution $\mu_3$ which can not be matched by $x$.
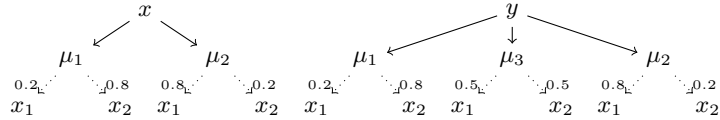
\begin{figure}
\begin{center} 
\begin{tikzpicture}[scale=0.65]
	\node (p) at (0,0) {$x$} ;
	\node (ppi1) at ($ (p) - (1.5,1) $) {$\mu_{1}$};
	\node (ppi2) at ($ (p) - (-1.5,1) $) {$\mu_{2}$};
	\node (p2) at ($ (ppi1) - (1,1) $) {$x_1$};
	\node (p3) at ($ (ppi1) - (-1,1) $) {$x_2$};
	\node (p4) at ($ (ppi2) - (1,1) $) {$x_1$};
	\node (p5) at ($ (ppi2) - (-1,1) $) {$x_2$};
	\path[->] (p) edge node [left] {{\tiny $\ $}} (ppi1);
	\path[->] (p) edge node [right] {{\tiny $\ $}} (ppi2);
	\path[->] (ppi1) [dotted]  edge node [left] {{\tiny $0.2$}} (p2);	
	\path[->] (ppi1) [dotted]  edge node [right] {{\tiny $0.8$}} (p3);	
	\path[->] (ppi2) [dotted]  edge node [left] {{\tiny $0.8$}} (p4);	
	\path[->] (ppi2) [dotted]  edge node [right] {{\tiny $0.2$}} (p5);	
\end{tikzpicture}
\begin{tikzpicture}[scale=0.65]
	\node (p) at (0,0) {$y$} ;
	\node (ppi1) at ($ (p) - (3,1) $) {$\mu_{1}$};
	\node (ppi2) at ($ (p) - (0,1) $) {$\mu_{3}$};
	\node (ppi3) at ($ (p) - (-3,1) $) {$\mu_{2}$};
	\node (p2) at ($ (ppi1) - (1,1) $) {$x_1$};
	\node (p3) at ($ (ppi1) - (-1,1) $) {$x_2$};
	\node (p4) at ($ (ppi2) - (1,1) $) {$x_1$};
	\node (p5) at ($ (ppi2) - (-1,1) $) {$x_2$};
	\node (p6) at ($ (ppi3) - (1,1) $) {$x_1$};
	\node (p7) at ($ (ppi3) - (-1,1) $) {$x_2$};

	\path[->] (p) edge node [left] {{\tiny $\ $}} (ppi1);
	\path[->] (p) edge node [right] {{\tiny $\ $}} (ppi2);
	\path[->] (p) edge node [right] {{\tiny $\ $}} (ppi3);
	\path[->] (ppi1) [dotted]  edge node [left] {{\tiny $0.2$}} (p2);	
	\path[->] (ppi1) [dotted]  edge node [right] {{\tiny $0.8$}} (p3);	
	\path[->] (ppi2) [dotted]  edge node [left] {{\tiny $0.5$}} (p4);	
	\path[->] (ppi2) [dotted]  edge node [right] {{\tiny $0.5$}} (p5);	
	\path[->] (ppi3) [dotted]  edge node [left] {{\tiny $0.8$}} (p6);	
	\path[->] (ppi3) [dotted]  edge node [right] {{\tiny $0.2$}} (p7);	
\end{tikzpicture}
\end{center}
\caption{Example of states $(x,y)$ not standard bisimilar.}
\label{example_pic_1}
\end{figure}

\subsection{Convex Bisimilarity}\label{convex_bisimilarity_section}
It has been argued by Segala (see, e.g.,  \cite{S95}) that standard bisimilarity is too strict a behavioral equivalence when PNTS's are used to model nondeterministic probabilistic programs/systems. In this setting, the nondeterminism in the system is supposed to model all the possible choices which can be made by, e.g., an external \emph{scheduler}. It is natural, however, to assume that schedulers can themselves use probabilistic methods to perform their choices. For instance the scheduling algorithm of an operating system might use probabilistic protocols to perform fair choices in an efficient way. Thus, a \emph{probabilistic scheduler} could choose to pick, from the state $x$, the successor distributions $\mu_1$ and $\mu_2$ with equal probability $\frac{1}{2}$, and consequentely reach states $x_1$ and $x_2$ with equal probabilities $\frac{1}{2}(0.2 + 0.8)=0.5$. Thus a scheduler, by choosing probabilistically between $\mu_1$ and $\mu_2$, can mimic the choice  of $\mu_3=\frac{1}{2}\mu_1 + \frac{1}{2}\mu_2$. 

\begin{definition}
Let $X$ be a set. A \emph{convex combination} of elements $\mu_i\in\mathcal{D}(X)$, for $i\in\{0,\dots, n\}$, is a probability distribution $\nu$ of the form $\nu(x)=\sum_i \lambda_i \mu_i(x)$, for $\lambda_i\in [0,1]$ and $\sum_{i}\lambda_i = 1$. 
The \emph{convex hull} of $A\!\in\!\mathcal{P}\mathcal{D}(X)$, denoted by $\textnormal{H}(A)$, is  the collection of all possible convex combinations of elements in $A$. The set $A$ is \emph{convex} if $A=\textnormal{H}(A)$. The entire set $\mathcal{D}(X)$, the emtpyset and the singletons $\{\mu\}$, for $\mu\in \mathcal{D}(X)$, are examples of convex sets.
\end{definition}




We are now ready to introduce 
the second notion of bisimilarity introduced by Segala \cite{S95,LS95}, which we refer to\footnote{The adjective \emph{probabilistic} is often adopted in the literature \cite{S95,SZG2011}. We prefer the adjective \emph{convex} which transparently reflects the mathematics behind the notion.} as \emph{convex bisimilarity},  based on the concept of probabilistic schedulers.

\begin{definition}[Convex Bisimulation]\label{def_convex_bisim}
Given a PNTS $(X,\alpha)$, a \emph{convex bisimulation} is an equivalence relation $E\subseteq X\times X$ such that if $(x,y)\in E$ then 
\begin{itemize}\item if $x\rightarrow_{C} \mu$ then there exists $\nu$ such that $y\rightarrow_{C} \nu $ and $\mu=_{E}\nu$,
\item if $y\rightarrow_{C} \nu$ then there exists $\mu$ such that $x\rightarrow_{C} \mu $ and $\mu=_{E}\nu$,
 \end{itemize}
where $x\rightarrow_{C}\mu$ if and only if $\mu\in \textnormal{H}(\alpha(x))$.
\end{definition}
 Thus convex bisimilarity is obtained by replacing $\rightarrow$ with $\rightarrow_{C}$ in Definition \ref{def_standard_bisim}. 
 Probabilistic schedulers take the place of ordinary schedulers and, by means of probabilistic choices, can simulate any convex combination of the reachable probability distributions. 
\begin{remark}\label{convex_algo}
Polynomial time algorithms for computing convex bisimilarity, based on linear programming, are studied in \cite{CS2002}.
\end{remark}
If probabilistic schedulers are assumed, Example \ref{example_pic_1} shows how, as observed by Segala in the first place, standard-bisimilarity is too strict because it distinguishes between states that, under the intended interpretation, ought to be identified. This fact does not mean that cocogruence (i.e., standard bisimilarity) is not the ``right'' notion of behavioral equivalence for $\mathcal{P}\mathcal{D}$-coalgebras. Rather, it suggests that $\mathcal{P}\mathcal{D}$-coalgebras do not precisely  model the class of systems we have in mind. We now briefly discuss how convex bisimilarity can be understood coalgebraically. 

\begin{definition}\label{functor_convex}
The endofuntor $\mathcal{P}_c\mathcal{D}$ on $\catset$ is defined as:
\begin{itemize}
\item $\mathcal{P}_c\mathcal{D}(X)\bydef\{ÊA \ | \ A\in \mathcal{P}\mathcal{D}(X) \textnormal{ and } A=\textnormal{H}(A)  \}$
\item $\big(\mathcal{D}(f)\big)(A)\bydef f[A] = \{ f[\mu] \ | \ \mu\in A \} $
\end{itemize}
where $f[\mu]$ is defined as in Definition \ref{functor_distribution}. 
 \end{definition}
 
 Hence $\mathcal{P}_c\mathcal{D}(X)$ is the collection of all convex  subsets of $\mathcal{D}(X)$. This means that $\mathcal{P}_c\mathcal{D}(X)\subseteq \mathcal{P}\mathcal{D}(X)$. Furthermore the action of $\mathcal{P}_c\mathcal{D}$ on morphisms $f$ is the same as that of $\mathcal{P}\mathcal{D}$ (i.e., $\mathcal{P}_c\mathcal{D}(f)$ is $\mathcal{P}\mathcal{D}(f)$ restricted to convex  sets) and indeed $f[A]$ is a convex set whenever $A$ is convex, as it is simple to verify.
 
 \begin{remark}
  The fact that $\mathcal{P}_c\mathcal{D}$ is a functor is well known. In fact more is true, and $\mathcal{P}_c\mathcal{D}$ carries a monad structure \cite{TKP2005}. This property has been extensively studied, especially in the field of domain theory \cite{TKP2005,GL2008,Mislove2000}, and recognized as important in the setting of probabilistic-nondeterminism. 
\end{remark}

It is clear that a $\mathcal{P}_c\mathcal{D}$-coalgebra $(X,\alpha)$ is just a particular kind of $\mathcal{P}\mathcal{D}$-coalgebra such that, for all $x\in X$, the set $\alpha(x)$ is convex. Furthermore the convex hull  operation gives us a natural way to convert a $\mathcal{P}\mathcal{D}$-coalgebra into a $\mathcal{P}_c\mathcal{D}$-coalgebra. 
\begin{proposition}\label{nat_transf_3}
The identity map $\textnormal{id}_{X}\!:\! \mathcal{P}_c\mathcal{D}(X)\!\rightarrow\!\mathcal{P}\mathcal{D}$, defined as $\textnormal{id}_{X}(A)=A$, is an injective natural transformation from $\mathcal{P}_c\mathcal{D}$ to $\mathcal{P}\mathcal{D}$. The  map $\textnormal{H}_{X}\!:\! \mathcal{P}\mathcal{D}(X)\!\rightarrow\!\mathcal{P}_c\mathcal{D}$, defined as $\textnormal{H}_{X}(A)=\textnormal{H}(A)$, is a surjective natural transformation from $\mathcal{P}\mathcal{D}$ to $\mathcal{P}_c\mathcal{D}$.
\end{proposition}
Transforming a $\mathcal{P}\mathcal{D}$-coalgebra (i.e., a PNTS) $(X,\alpha)$ into the $\mathcal{P}_c\mathcal{D}$-coalgebra $(X,\textnormal{H}_X\!\circ\! \alpha)$ precisely corresponds to the substitution of the arrow relation $\rightarrow$ in Definition \ref{def_standard_bisim} with the relation $\rightarrow_{C}$ in Definition \ref{def_convex_bisim}. It is now straightforward to verify that an equivalence relation $E\subseteq X\times X$ is  a convex bisimulation  in the PNTS $(X,\alpha)$ if and only if $E$ is a cocongruence (i.e., standard bisimilarity) in $(X,\textnormal{H}_X\!\circ\! \alpha)$. 
The process of transforming a $\mathcal{P}\mathcal{D}$-coalgebra $(X,\alpha)$ into a $\mathcal{P}_c\mathcal{D}$-coalgebra generally causes a loss of information. For instance, the two systems of Figure \ref{example_pic_1} are both mapped to the same $\mathcal{P}_c\mathcal{D}$-coalgebra. This transformation process is really a quotient operation ($\textnormal{H}_X$ is surjective) induced by the \emph{behavioral equation} $A\!\approx\! B$ whenever $\textnormal{H}(\textnormal{A})\!=\!\textnormal{H}(B)$, for all $A,B\in\mathcal{P}\mathcal{D}(X)$,  capturing the the behavior of probabilistic schedulers. 
\begin{remark}\label{working_with_convex}
Although $\mathcal{P}_c\mathcal{D}$-coalgebras are the models naturally corresponding to convex bisimilarity, we can always work, concretely, with ordinary PNTS's (i.e., $\mathcal{P}\mathcal{D}$-coalgebras) $(X,\alpha)$, perhaps represented as finite graphs, and tacitly replace $\alpha$ with $\textnormal{H}_X\circ\alpha$ (i.e., $\rightarrow$ with $\rightarrow_C$). This is convenient since, generally, the convex hull of a finite set is uncountable.
\end{remark}

The discussio carried out in this subsection serves to clarify that no \emph{a priori} categorical or coalgebraic argument exists supporting a notion of behavioral equivalence
in favor of another, when modeling computing systems. Coalgebra provides, e.g., the mathematically deep notion of cocongruence for a functor, but the choice of an appropriate functor is part of the modeling process.

\section{Upper Expectation Bisimilarity}\label{section_UE_bisim_intro}
We saw how convex bisimilarity naturally arises from the observation that schedulers may make probabilistic choices. We also discussed how it can be understood coalgebraically in terms of cocongruences on $\mathcal{P}_c\mathcal{D}$-coalgebras. However it is not possible to claim, on the sole basis of these facts, that convex bisimilarity is a convenient notion of behavioral equivalence for PNTS's. Probabilistic schedulers constitute a good reason to consider two convex bisimilar states as behaviorally equivalent. But it is not clear why one might want to distinguish between two states that are not convex bisimilar. We illustrate the problem by means of the simple example of Figure \ref{example_pic_2}. As usual we assume the three states $x_1$, $x_2$ and $x_3$ to be observationally distinct (cf. Footnote \ref{obs_diff}).
The two states $x$ and $y$ are not convex bisimilar because $\mu_3$ is not a convex combination of $\mu_1$ and $\mu_2$.
\begin{figure}
\begin{center}
\begin{tikzpicture}[scale=0.65]
	\node (p) at (0,0) {$x$} ;
	\node (d1) at ($ (p) - (1.5,1) $) {$\mu_{1}$};
	\node (d2) at ($ (p) - (-1.5,1) $) {$\mu_{2}$};
	\node (p1) at ($ (d1) - (1,1) $) {$x_1$};
	\node (p2) at ($ (d1) - (0,1) $) {$x_2$};
	\node (p3) at ($ (d1) - (-1,1) $) {$x_3$};
	\node (p4) at ($ (d2) - (1,1) $) {$x_1$};
	\node (p5) at ($ (d2) - (0,1) $) {$x_2$};
	\node (p6) at ($ (d2) - (-1,1) $) {$x_3$};
	\path[->] (p) edge node [left] {{\tiny $\ $}} (d1);
	\path[->] (p) edge node [right] {{\tiny $\ $}} (d2);
	\path[->] (d1) [dotted]  edge node [left] {{\tiny $0.3$}} (p1);	
	\path[->] (d1) [dotted]  edge node  {{\tiny $0.3$}} (p2);	
	\path[->] (d1) [dotted]  edge node [right] {{\tiny $0.4$}} (p3);	
	\path[->] (d2) [dotted]  edge node [left] {{\tiny $0.5$}} (p4);
	\path[->] (d2) [dotted]  edge node  {{\tiny $0.4$}} (p5);		
	\path[->] (d2) [dotted]  edge node [right] {{\tiny $0.1$}} (p6);	
\end{tikzpicture}
\begin{tikzpicture}[scale=0.65]
	\node (p) at (0,0) {$y$} ;
	\node (d1) at ($ (p) - (2.5,1) $) {$\mu_{1}$};
	\node (d2) at ($ (p) - (-2.5,1) $) {$\mu_{2}$};
	\node (d3) at ($ (p) - (0,2) $) {$\mu_{3}$};
	\node (p1) at ($ (d1) - (1,1) $) {$x_1$};
	\node (p2) at ($ (d1) - (0,1) $) {$x_2$};
	\node (p3) at ($ (d1) - (-1,1) $) {$x_3$};
	\node (p4) at ($ (d2) - (1,1) $) {$x_1$};
	\node (p5) at ($ (d2) - (0,1) $) {$x_2$};
	\node (p6) at ($ (d2) - (-1,1) $) {$x_3$};
	\node (q1) at ($ (d3) - (1,1) $) {$x_1$};
	\node (q2) at ($ (d3) - (0,1) $) {$x_2$};
	\node (q3) at ($ (d3) - (-1,1) $) {$x_3$};
	
	\path[->] (p) edge node [left] {{\tiny $\ $}} (d1);
	\path[->] (p) edge node [right] {{\tiny $\ $}} (d2);
	\path[->] (p) edge node [right] {{\tiny $\ $}} (d3);
	\path[->] (d1) [dotted]  edge node [left] {{\tiny $0.3$}} (p1);	
	\path[->] (d1) [dotted]  edge node  {{\tiny $0.3$}} (p2);	
	\path[->] (d1) [dotted]  edge node [right] {{\tiny $0.4$}} (p3);	
	\path[->] (d2) [dotted]  edge node [left] {{\tiny $0.5$}} (p4);
	\path[->] (d2) [dotted]  edge node  {{\tiny $0.4$}} (p5);		
	\path[->] (d2) [dotted]  edge node [right] {{\tiny $0.1$}} (p6);	
	\path[->] (d3) [dotted]  edge node [left] {{\tiny $0.4$}} (q1);
	\path[->] (d3) [dotted]  edge node  {{\tiny $0.3$}} (q2);		
	\path[->] (d3) [dotted]  edge node [right] {{\tiny $0.3$}} (q3);	
	
	\end{tikzpicture}
\end{center}
\caption{Example of states $(x,y)$ not convex bisimilar.}
\label{example_pic_2}
\end{figure}
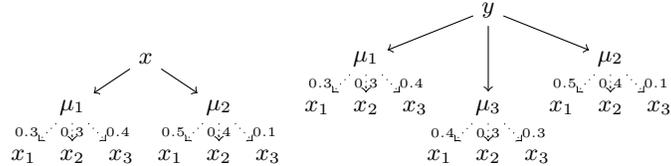
It is not simple, however, to find a \emph{concrete}\footnote{This is deliberately a vague adjective which could be seconded by, \emph{experimental}, \emph{operational}, \emph{testing-based}, \dots} reason to distinguish between the two states. As a matter of fact, it has been proved in \cite{SZG2011} that the states $x$ and $y$ satisfy the same properties formulated in the expressive logic $\textnormal{PCTL}^*$ of \cite{BA1995}.
\begin{remark}
While modal logics, carefully crafted to capture convex bisimilarity (and  even standard bisimilarity) can be defined \cite{LS95,S95,DvG2010,DvGHM2009},  it is certainly interesting to look at the distinguishing power of popular temporal logics for PNTS's, (of which $\textnormal{PCTL}^*$ is a main example)  capable of expressing branching properties of probabilistic concurrent systems useful in practice (see, e.g., \cite{PRISM4}).
\end{remark}

To explain why the two states $x$ and $y$ have indeed similar behavior, we now introduce a simple experimental scenario. Suppose we are allowed to make repeated experiments (in the sense of Milner's push-buttons metaphor \cite{M80}) on the PNTS's of Figure \ref{example_pic_1}. After $n\!\to\!\infty$ experiments (that is, for  $n$ as big as desired) we observe that an event $A$ (e.g., $A=\{x_1 \}$, representing the occurrence of terminal state $x_{1}$) happened $m$ times. We can then make the following reasonable\footnote{This is of course based on the common frequentist interpretation of probabilities as limits of relative frequencies in a large number of trials.} assessment:  the PNTS \emph{may} exhibit behavior $A$ (i.e., end up in $x_1$) with probability $\frac{m}{n}$.  
It seems then natural to stipulate that two states $x$ and $y$ are equivalent if, for each event $A$, the state $x$ can exhibit behavior (event) $A$ with probability (frequency) $\lambda$ if and only if $y$ can. This leads to the following definition.

\begin{definition}
Let $X$ be a set and $S\in\mathcal{P}\mathcal{D}(X)$ a set of probability distributions on $X$. The \emph{upper probability functional} $up_S\!:\!\mathcal{P}(X)\!\rightarrow\! [0,1]$, mapping subsets $A\subseteq X$ to real numbers in $[0,1]$, is defined as  the supremum value $up_{S}(A)=\bigsqcup_{\mu\in S}\mu(A)$.
\end{definition}
The adjective \emph{functional} is motivated by the fact that $\mathcal{P}(X)$ can be seen as the space of characteristic functions $X\rightarrow\{0,1\}$. Thus $up_S$ maps functions (predicates) to real numbers.
\begin{definition}
Given a PNTS $(X,\alpha)$, an \emph{upper probability (UP) bisimulation} is an equivalence relation $E\subseteq X\times X$ such that if $(x,y)\in E$ then the equality 
$up_{\alpha(x)}(A) = up_{\alpha(y)}(A)$ holds for all sets $A\subseteq X$ which are unions of $E$-equivalence classes.
\end{definition}

\begin{remark}
The choice of considering upper probabilities seems one-sided, but one could equally well choose lower probabilities ($lp_S$) observing that $lp_S(A)\!=\!1-up_S(X\setminus A)$. 
\end{remark}

The idea behind the definition is simple. If $x$ and $y$ are UP-bisimilar, than if $x$ can exhibit event $A$ with probability $\lambda$ then also $y$ can, and \emph{viceversa}. At the same time, if $x$ and $y$ are not UP-bisimilar, there exists an event $A$ such that an agent, by performing on $x$ and $y$ a sufficiently large number of repeated experiments looking after the occurrences of event $A$, \emph{may} observe a probability of $A$ on $x$ that \emph{can not} be achieved in $y$.
The restriction on $A$ being union of $E$-equivalence classes is natural, since we wish to identify (i.e., apply $x\mapsto [x]_E$)  $E$-related (i.e., UP-bisimilar) states.

We introduced UP-bisimilarity to understand in what sense the behavior of the states $x$ and $y$ of Figure \ref{example_pic_2} are similar. It is now simple to verify that $x$ and $y$ are UP-bisimilar. Thus $x$ and $y$ can be not distinguished by means of experiments.  As it turns out, UP-bisimilarity has been recently derived (with different names) in  \cite{BdNM2012} and \cite{SZG2011}, albeit following a different conceptual path. It is shown in \cite{SZG2011} that UP-bisimilarity is strictly coarser than the equivalences induced by the logic $\textnormal{PCTL}^*$ and  its weaker fragment $\textnormal{PCTL}$. This is a significant drawback of UP-bisimilarity as one definitely wants to distinguish between states that can be separated by useful temporal logics such as  $\textnormal{PCTL}$ and $\textnormal{PCTL}^*$. However we derived UP-bisimilarity following a straightforward testing metaphor, which we now refine to get a better behaved alternative.

Suppose indeed that in the experiments carried on PNTS's, one is not just allowed to observe the occurrences of events and thus, after enough ($n\to \infty$) experiments, estimate their probabilities by means of relative frequencies. Rather, one can associate a real valued information $r_i\in\mathbb{R}$ to the outcome of each experiment $i$ and, after $n$ experiments, observe the average value $\frac{1}{n}\sum_{i=0}^n r_i$. This new scenario is better explained by a simple example. Consider again the states $x$ and $y$ of Figure \ref{example_pic_2} and consider the function $g:\{x_1,x_2,x_3\}\rightarrow\mathbb{R}$ defined as $g(x_1)=60$, $g(x_2)=0$ and $g(x_3)=50$. The function $g$ represents the experiment in the sense that if, after letting the scheduler choose a transition (i.e., pushing the button, in Milner's metaphor), the state $x_i$ is reached,  then the real number $g(x_i)$ is registered as result. Thus, for instance, if $\{ x_1,x_2,x_3,x_2,x_1\}$ was the outcome of five experiments, the numerical sequence $\{60,0,50,0,60\}$ and its average $34=\frac{170}{5}$ would be our observation. 
\begin{definition}
Let $X$ be a set, $\mu\in\mathcal{D}(X)$ and $f:X\rightarrow\mathbb{R}$. The \emph{expected value} of $f$ under $\mu$, written $E_\mu(f)$, is defined as $E_\mu(f)=\sum_x \mu(x)f(x)$.
\end{definition}

Note that the expected values of $g$ under $\mu_1$, $\mu_2$ and $\mu_3$ are $38$, $35$ and $39$, respectively. This readily means that, for $n\to\infty$, the average resulting from experiments $g$ on state $x$ is \emph{necessarily} smaller or equal than $38$, while in state $y$ it \emph{can be} strictly greater than $38$ (and at most $39$). Thus, it is \emph{possible} that an agent, by means of a sufficiently large number of repeated experiments $g$, \emph{may} be able to distinguish between the two states $x$ and $y$. This discussion leads to the following definitions.
\begin{definition}\label{ue_functional}
Let $X$ be a set and $S\in\mathcal{P}\mathcal{D}(X)$ a set of probability distributions on $X$. The \emph{upper expectation functional} $ue_S:(X\rightarrow\mathbb{R})\rightarrow \mathbb{R}$, mapping functions $X\rightarrow \mathbb{R}$ to real numbers, is defined as $ue_{S}(f)=\bigsqcup_{\mu\in S}E_\mu(f)$.
\end{definition}
\begin{definition}\label{def_UE_bisimilarity}
Given a PNTS $(X,\alpha)$, an \emph{upper expectation (UE) bisimulation} is an equivalence relation $E\subseteq X\times X$ such that if $(x,y)\in E$ then the equality 
$ue_{\alpha(x)}(f) = ue_{\alpha(y)}(f)$ holds for all \emph{$E$-invariant} functions $f:X\rightarrow \mathbb{R}$, i.e., such that if $(z,w)\in E$ then $f(z)=f(w)$.
\end{definition}
\begin{remark}\label{remark_dealfaro1}
Up to minor modifications, the notion of UE bisimilarity already appeared in the literature \cite{deAlfaro2008} under the name of \emph{game bisimilarity}, in the (significantly) different setting of two-player stochastic (concurrent) games. See Remark \ref{remark_dealfaro2} below for further discussion.
\end{remark}
We restrict to $E$-invariant experiments $f$ following the usual idea that $E$-related (i.e., UE-bisimilar) states ought to be identified.  Note how UP-bisimilarity is obtained by restricting the set of $\mathbb{R}$-valued functions $f$ to $\{0,1\}$-valued functions $g\!:\!X\rightarrow\!\{0,1\}$, in the definition of UE-bisimulation. Indeed note that  $g$ is $E$-invariant if and only if  $g$ is the characteristic function of some set $Y\subseteq X$ which is union of $E$-equivalence classes. This observation, together with the fact that the two states $x$ and $y$ of Figure \ref{example_pic_2} are UP-bisimilar but not UE-bisimilar,  reveals that $\mathbb{R}$-valued experiments (modeled by $f\!:\!X\rightarrow\!\mathbb{R}$) generally provide more information than just observations of events (modeled by $g\!:\!X\rightarrow\!\{0,1\}$).

The definition of UE-bisimulation gives reasons for distinguishing states based on the existence of a witnessing experiment $f$ which \emph{may} have an expected outcome in one state which \emph{cannot} be matched by the other state. The dual \emph{modalities} may/must (can/cannot) considered in our testing scenario are, of course, well known key concepts in  classical bisimulation theory, modal logics and concurrency theory \cite{M80,BdRVModal}.

\subsection{Relation between Convex and UE-bisimilarity}
UE-bisimilarity arises naturally from the simple testing scenario discussed above and, as we shall discuss in Section \ref{real_valued_section}, enjoys a remarkable natural connection with real-valued modal logics. It is thus worth to develop its theory and compare it with that of convex bisimilarity. In this section we show that the two notions coincide on a wide class of systems by means of an alternative characterization of UE-bisimilarity.
Recall that $\mathcal{P}_c\mathcal{D}$-coalgebras (cf. Remark \ref{working_with_convex}) can be thought as $\mathcal{P}\mathcal{D}$-coalgebras modulo the behavioral equation $A\!\approx \!B$ if $\textnormal{H}(A)\!=\!\textnormal{H}(B)$, for $A,B\in\mathcal{P}\mathcal{D}(X)$. Following the same idea, to understand UE-bisimilarity coalgebraically one needs to consider the behavioral equation $A\!\approx \!B$ if $ue_A\!=\! ue_B$ (pointwise equality). As it turns out, the two equations coincide under very mild conditions. In rest of the paper we restrict attention to a fairly simple (yet important) class of PNTS's, as this greatly simplifies the technical development.
\begin{convention}\label{conv_1}
We restrict attention to PNTS's having a finite state space, i.e., $\mathcal{P}\mathcal{D}$-coalgebras $(X,\alpha\!:\!X\!\rightarrow\! \mathcal{P}\mathcal{D}(X))$ having a finite carrier set $X\!=\!\{x_1,\dots,x_n\}$. This allows one to view the space of functions $X\!\rightarrow\!\mathbb{R}$ as the Euclidean space $\mathbb{R}^n$ and each $\mu\!\in\!\mathcal{D}(X)$ as the $n$-dimensional vector $\mu=[\mu(x_1),\dots,\mu(x_n)]$. Note that $\mathcal{D}(X)$ is a \emph{closed} and \emph{bounded} (hence \emph{compact}) subset of $\mathbb{R}^n$ and that $\alpha(x)$, for $x\in X$, can be infinite.
\end{convention}
In Section \ref{sec_infinity} we briefly discuss how our results can be generalized to infinite systems at the cost of mathematical complications since infinite dimensional topological spaces need to be considered.

The following is a known result in optimization theory and applied statistics (see, e.g., \cite[\S 10.2]{Huber}, \cite{Walley1991} and \cite{HP2007}).

\begin{theorem}[\cite{Huber}]\label{theorem_UE=convex}
Let $X$ be a finite set and $A,B\in\mathcal{P}\mathcal{D}(X)$. Then 
 $ue_A\!=\!ue_B$ if and only if $\textnormal{cl}(\textnormal{H(A)})\!=\!\textnormal{cl}(\textnormal{H(B)})$,
where $\textnormal{cl}(C)$ denotes the \emph{topological closure} of the set $C$.
\end{theorem}
It is a standard result in linear algebra that the closure of a convex set is itself convex (see, e.g., \cite[\S 8.4]{LAX}). For this reason the set $\textnormal{cl}(\textnormal{H(A)})$ is called the \emph{closed convex hull} of $A$. In what follows we denote with $\overline{\textnormal{H}}$ the operation $A\mapsto \textnormal{cl}(\textnormal{H(A)})$.

We can use the result of Theorem \ref{theorem_UE=convex} to prove the following alternative characterization of UE-bisimulation (cf. Definition \ref{def_convex_bisim}).

\begin{theorem}\label{def_convexclosed_bisim}
Given a PNTS $(X,\alpha)$, an equivalence relation $E\subseteq X\times X$ is a UE-bisimulation iff for all $(x,y)\in E$, 
\begin{itemize}\item if $x\!\rightarrow_{CC}\! \mu$ then there exists $\nu$ such that $y\!\rightarrow_{CC}\! \nu $ and $\mu\!=_{E}\!\nu$,
\item if $y\!\rightarrow_{CC}\! \nu$ then there exists $\mu$ such that $x\!\rightarrow_{CC}\! \mu $ and $\mu\!=_{E}\!\nu$,
 \end{itemize}
where $x\!\rightarrow_{CC}\!\mu$ if and only if $\mu\in \overline{\textnormal{H}}(\alpha(x))$.
\end{theorem}
\begin{proof}
For $A\in\mathcal{P}\mathcal{D}(X)$, denote with $A/E\!\in\!\mathcal{P}\mathcal{D}(X/E)$ the set $A/E=\{  q[\mu] \ | \ \mu\in A\}$ (cf. discussion after Definition \ref{def_standard_bisim}), where $q$ is the quotient map of $E$. We can now rephrase the assertion of the theorem as follows. The relation $E$ is a UE-bisimulation iff for every $(x,y)\!\in\! E$ it holds that $\overline{\textnormal{H}}(\alpha(x))/E \!=\!\overline{\textnormal{H}}(\alpha(y))/E$.

Observe that the map $\mu\mapsto q[\mu]$ is linear, hence continuous since $\mathbb{R}^n$ is finite dimensional. This implies that, for every convex closed set $A\!\in\!\mathcal{P}\mathcal{D}(X)$, the set 
 $\overline{\textnormal{H}}(A)/E$ is a convex closed set in $\mathcal{P}\mathcal{D}(X/E)$. 
Note that there exists a one-to-one correspondence between $E$-invariant functions $f\!:\!X\!\rightarrow\! \mathbb{R}$ and arbitrary functions $g\!:\!X/E\!\rightarrow\!\mathbb{R}$. The definition of  UE-bisimilarity as in Definition \ref{def_UE_bisimilarity} can then be rephrased as: $(x,y)\in E$ implies $ue_{\alpha(x)/E}=ue_{\alpha(y)/E}$. The result then follows by Theorem \ref{theorem_UE=convex}.
\end{proof}

Thus UE-bisimilarity can be obtained by replacing $\rightarrow_{C}$ with $\rightarrow_{CC}$ in Definition \ref{def_convex_bisim} of convex bisimilarity.   As a consequence UE-bisimilarity can be strictly coarser than convex bisimilarity. The following lemma, however, reveals that the two notions coincide for a wide class of PNTS's.
\begin{proposition}\label{closedness_condition}
Let $(X,\alpha)$ be a PNTS such that $\alpha(x)$ is closed for all $x\in X$. Then $E\subseteq X\times X$ is a convex bisimulation if and only if it is a UE-bisimulation.
\end{proposition}
\begin{proof}
If $\alpha(x)\!\subseteq\! \mathcal{D}(X)$ is closed then it is compact since $\mathcal{D}(X)$ is bounded. The convex hull of a compact set in $\mathbb{R}^n$ is itself closed and compact. Thus $\textnormal{H}(\alpha(x))\!=\!\overline{\textnormal{H}}(\alpha(x)))$ and the result follows.
\end{proof}
Restricting to PNTS's of this kind can hardly be seen as a limitation in concrete applications. First, every finite set is closed, thus every PNTS representable as a finite graph satisfies Convention \ref{conv_1} and the closedness condition of Proposition \ref{closedness_condition}. Furthermore every NTS and every Markov process (see \cite{Sokalova2011} and definition below) can be seen as PNTS's satisfying the closedness condition.
\begin{definition}
A (discrete) \emph{Markov process} (or Markov \emph{chain}) is a $\textnormal{MP}$-coalgebra of the functor $\textnormal{MP}(X)=\{*\}+\mathcal{D}(X)$.
\end{definition} 
 
\begin{proposition}\label{NTS_as_PNTS}
The set-indexed collection of maps $\eta_X\!:\!\textnormal{MP}(X)\rightarrow\mathcal{P}\mathcal{D}(X)$ and
$\gamma_X\!:\!\mathcal{P}(X)\rightarrow\mathcal{P}\mathcal{D}(X)$
defined as:
\begin{itemize}
\item $\eta_{X}(*)=\emptyset\ $ and $\ \eta_{X}(\mu)=\{\mu\}$
\item $\gamma_{X}( \{ x_i\}_{i\in I}) = \overline{\textnormal{H}}( \{Ê\delta(x_i)\}_{i\in I})$, with $\delta(x_i)$ as in Definition \ref{functor_distribution}
\end{itemize}
are injective natural transformations.
\end{proposition}
Thus we shall consider every NTS and Markov Process as the (unique) corresponding PNTS with the closedness property.

\begin{remark}
Due to the lack of space, we just mention in this remark that it possible to define the functor $\mathcal{P}_{cc}\mathcal{D}$ mapping $X$ to the space of convex compact closed subsets of $\mathcal{D}(X)$ and consider the expected natural transformations $\textnormal{id}\!:\! \mathcal{P}_{cc}\mathcal{D}\rightarrow \mathcal{P}_c\mathcal{D}$ and $\textnormal{cl}\!:\!\mathcal{P}_c\mathcal{D}\!\rightarrow\! \mathcal{P}_{cc}\mathcal{D}$. It is then simple to verify that the notion of UE-bisimilarity coincides with that of cocongruence on $\mathcal{P}_{cc}\mathcal{D}$-coalgebras (cf. Remark \ref{working_with_convex} and associated discussions). Note that the natural transformations of Proposition \ref{NTS_as_PNTS} restrict to type $\textnormal{MP}\rightarrow\mathcal{P}_{cc}\mathcal{D}$ and $\mathcal{P}\rightarrow\mathcal{P}_{cc}\mathcal{D}$, respectively.
\end{remark}

As discussed in Section \ref{def_convex_bisim},  the notion of probabilistic scheduler gives good reasons for considering two convex bisimilar states as behaviorally equivalent. On the other hand UE-bisimilarity provides witnesses (experiments) $f\!:\!X\!\rightarrow\!\mathbb{R}$ which can be used to distinguish states that are not UE-bisimilar (cf. example of Figure \ref{example_pic_2}). Thus the two \emph{a posteriori} equivalent (under the mild  closedness assumption) viewpoints complement each other in a nice way. 


\begin{remark}\label{remark_dealfaro2}
As already mentioned in Remark \ref{remark_dealfaro1}, in the important work of \cite{deAlfaro2008}, game bisimilarity is discovered as the kernel 
of a behavioral (pseudo)metric $d$ induced by the quantitative $[0,1]$-valued logic qL$\mu^\ominus$ (which, as we will discuss in Section \ref{logic_mu_calculi}, is a weak logic not capable of encoding, e.g., the  logic PCTL of \cite{BA1995}). The authors of \cite{deAlfaro2008} argue in favor of game bisimilarity on the basis of the naturalness of the logic qL$\mu^\ominus$. Our explanation in terms of real-valued experiments is perhaps useful in clarifying and further motivating this notion using logic-free arguments based on the simple (and, as we will discuss in Section \ref{logic_mu_calculi}, more general) metaphor of $\mathbb{R}$-valued experiments. No connection between UE-bisimilarity and convex bisimilarity is discussed in \cite{deAlfaro2008}.
\end{remark}

\section{Real Valued Modal Logics}\label{real_valued_section}
The result of Theorem \ref{theorem_UE=convex} states that the closed convex hull operation $A\mapsto \overline{\textnormal{H}}(A)$ and the upper expectation
$A\mapsto ue_A$ are essentially the same operation.
\begin{proposition}[\cite{Huber}]
Let $X$ be a finite set and $A\in\mathcal{P}\mathcal{D}(X)$. Then $\{ \mu \in \mathcal{D}(X)\ | \ \forall f\!:\!X\!\rightarrow\!\mathbb{R}. \big( E_\mu(f)\!\leq\! ue_A(f) \big)\}\!=\! \overline{\textnormal{H}}(A)$.
\end{proposition}
\noindent
Thus from the functional $ue_A$ it is possible to construct $ \overline{\textnormal{H}}(A)$ and, by Theorem \ref{theorem_UE=convex}, from $\overline{\textnormal{H}}(A)$ one obtains $ue_A\!=\!ue_{\overline{\textnormal{H}}(A)}$. Hence, we can look at the transition map $\alpha$ of a PNTS $(X,\alpha)$, with $\alpha(x)$ convex closed for all $x\in X$ (i.e., a $\mathcal{P}_{cc}\mathcal{D}$-coalgebra, cf. Remark \ref{working_with_convex}), both as a function  $(x\mapsto \alpha(x))$ of type $X\rightarrow \mathcal{P}_{cc}\mathcal{D}(X)$ and as a function  ($x\mapsto ue_{\alpha(x)}$) of type $X\rightarrow \big((X\rightarrow\mathbb{R})\rightarrow\mathbb{R}\big)$. Equivalently (by currying) the transition map $\alpha$ can seen as the function transformer $\Diamond_\alpha: (X\rightarrow\mathbb{R})\rightarrow(X\rightarrow\mathbb{R})$ defined as:
\begin{equation}\label{definition_diamond}
\big( \Diamond_{\alpha} f\big)  (x)= ue_{\alpha(x)}(f) \stackrel{\textnormal{Def \ref{ue_functional}}}{=} \displaystyle \bigsqcup_{x\rightarrow \mu}E_\mu(f)
\end{equation}
 It is remarkable here that the function transformer $\Diamond_{\alpha}$ happens to coincide with the interpretation (i.e., semantics) of the \emph{diamond modality} in all quantitative logics for PNTS's in the literature \cite{HM96,MM07,AM04,MIO2012b,DGJP2002}. While it is obvious that the PNTS $(X,\alpha)$ induces $\Diamond_\alpha$ (just as in the definition), the fact that from $\Diamond_\alpha$ one can reconstruct the $\mathcal{P}_{cc}\mathcal{D}$-coalgebra $(X,\alpha)$ is far from clear. Some of the consequences of this observation are explored in Section \ref{logic_R}.

$\mathbb{R}$-valued  modal logics are based on the fundamental idea of replacing ordinary Boolean predicates (i.e., $\{0,1\}$-valued functions) with quantitative ones. Thus, the semantics of a formula $\phi$, interpreted on a PNTS $(X,\alpha)$, is a function $\sem{\phi}\!:\!X\!\rightarrow\!\mathbb{R}$ and, in particular, $\sem{\Diamond\phi}\!=\!\Diamond_\alpha(\sem{\phi})$. Following our discussions, \emph{formulas} can then be thought of as \emph{experiments} and the value $\sem{\Diamond\phi}(x)$ as the maximal (limit) expected value of experiment $\phi$ performed after ``pushing the button'' (in Milner's terminology) at $x$. 
\begin{example}\label{example_formulas_as_experiments}
E.g., the experiment distinguishing the states $x$ and $y$ of Figure \ref{example_pic_2} corresponds to the formula $\Diamond \phi$, for some $\phi$ crafted in such a way that $\sem{\phi}(x_1)=60$, $\sem{\phi}(x_2)=0$ and $\sem{\phi}(x_3)=50$.
\end{example}
The many concrete $\mathbb{R}$-valued modal logics appearing in the literature are obtained by considering other connectives to be used in combination with $\Diamond$. For example the constant $\underline{1}$ ($\sem{\underline{1}}(x)\!=\!1$) and the  connective $\sqcup$ ($\sem{\phi\sqcup \psi}(x)\!=\!\max\{Ê\sem{\phi}(x), \sem{\psi}(x)\}$) are considered in all the logics we are aware of. Different choices of connectives have distinct advantages over each other in terms of, e.g., model checking complexity, expressivity, game semantics, compositional reasoning methods, \emph{etc}. All $\mathbb{R}$-valued modal logics we are aware of are \emph{sound} with respect to UE-bisimilarity:
\begin{definition}[Soundness]\label{soundness_def}
An $\mathbb{R}$-valued modal logic is \emph{sound} if, for every PNTS $(X,\alpha)$, UE-bisimulation $E$ and formula $\phi$ it holds that if $(x,y)\!\in\! E$ then $\sem{\phi}(x)=\sem{\phi}(y)$, i.e., $\sem{\phi}$ is $E$-invariant.
\end{definition}
The following is a natural \emph{desideratum} for a modal logic:
\begin{definition}[Strong Completeness]\label{strong_completeness_def}
A $\mathbb{R}$-valued modal logic is \emph{strongly complete} if, for every $(X,\alpha)$, the  set of functions $\sem{\phi}:X\rightarrow\mathbb{R}$, with $\phi$ ranging over formulas, is dense\footnote{We consider the uniform metric on $X\rightarrow\mathbb{R}$, i.e., the one induced by the sup norm. Thus $\mathcal{F}$ is dense in  $X\rightarrow\mathbb{R}$ if for every $f:X\rightarrow\mathbb{R}$ and $\epsilon >0$, there exists $g\in\mathcal{F}$ such that $\max_{x\in X}{ | f(x)-g(x)| <\epsilon}.$} in the set $X\rightarrow\mathbb{R}$ of functions that respect UE-bisimilarity, i.e., that are $E$-invariant for every UE-bisimulation $E$.
\end{definition}
Thus a logic is strongly complete if every experiment $X\rightarrow\mathbb{R}$ (which can not separate UE-bisimilar states) can be approximated, with an arbitrary level of precision, by (the denotation of) some formula $\phi$. Note that we consider approximations since a logic typically consists of a countable set of formulas, while $X\rightarrow\mathbb{R}$ is uncountable. Strong completeness is a natural notion and, as we will discuss in the next sections, a mathematically well behaved one. Note that, another weaker form of completeness could be considered, simply requiring that, whenever $\sem{\phi}(x)=\sem{\phi}(y)$ for all $\phi$, then $x$ and $y$ are UE-bisimilar. As mentioned earlier (cf. Remark \ref{remark_dealfaro2}) the logic qL$\mu^\ominus$ of \cite{deAlfaro2008} is weakly complete. 

We will show in Section \ref{logic_mu_calculi} that the probabilistic $\mu$-calculi of \cite{MioThesis} are sound and strongly complete. Before proving this, however, we develop the theory of a novel modal logic in Section \ref{logic_R} designed on the basis of first principles coming from deep results in linear algebra and functional analysis discussed in Section \ref{sec_representation_theorems}.

\begin{remark}
It is important to appreciate how we derived the core of a $\mathbb{R}$-valued logic for PNTS's (i.e., the interpretation of the basic modality $\Diamond$) from the very elementary and concrete observation (motivating UE-bisimilarity) that $\mathbb{R}$-valued experiments on PNTS's are useful and, due to their greater observational power, should replace ordinary $\{0,1\}$-observations (cf. example of Figure \ref{example_pic_2}).
\end{remark}

\subsection{Representation Theorems}\label{sec_representation_theorems}
The function space $X\!\rightarrow\!\mathbb{R}$ is an object of supreme importance in linear algebra and functional analysis \cite{LAX}. It is the real vector space (with sums and scalar multiplication defined pointwise), often denoted by $C(X)$, of all $\mathbb{R}$-valued continuous\footnote{Every function $X\rightarrow\mathbb{R}$ is continuous since $X$ is assumed to be finite and endowed with the discrete topology.} functions on $X$ normed\footnote{\label{remark_unique_topology}Recall that every norm of a finite dimensional space induces the same topology. In other words, the notion of closedness is norm-invariant in $\mathbb{R}^n$.} by the \emph{sup norm} (also known as $L^\infty$ norm \cite{LAX}), defined as $\| f \|_\infty \!=\! \max_{x}\! |f(x)|$. The subspace of $C(X)\!\rightarrow\!\mathbb{R}$ consisting of \emph{linear functionals} (i.e., such that  $F(\lambda_1 f_1+\lambda_2 f_2)=\lambda_1F(f_1)+\lambda_2F(f_2)$) is known as the \emph{continuous linear dual of $C(X)$}, and denoted by $C(X)^*$. The following is the finite-dimentional variant of the Riesz Representation Theorem, a result of fundamental importance in analysis and measure theory (see, e.g., \cite[\S A.1]{LAX}).
\begin{theorem}[Riesz Representation Theorem]\label{riesz_representation}
Let $X$ be a finite set and $\mu\!\in\!\mathcal{D}(X)$. Then $E_\mu$ is a monotone (i.e., if $f\sqsubseteq g$ then $E_\mu(f)\leq E_\mu(g) $) linear functional. Dually, every monotone linear functional $F\!\in\! C(X)^*$ such that $F(\underline{1})\!=\!1$ (where $\underline{1}(x)=1)$ is of the form $E_\mu$ for some $\mu\!\in\!\mathcal{D}(X)$
\end{theorem}
Since we are interested in convex closed sets of probability distributions, rather then probabity distributions, we now consider the following variant of  Riesz's theorem.
\begin{theorem}[\cite{Huber}]\label{convex_representation}
Let $X$ be a finite set and $A\!\in\!\mathcal{P}\mathcal{D}(X)$ be a convex closed and nonempty set. Then $ue_A$ is a monotone functional $F\!:\!C(X)\!\rightarrow\!\mathbb{R}$ which is subadditive (i.e., $F(f+g)\leq F (f) + F(g)$), positively affinely homogeneous (i.e., for all $\lambda_1\geq 0$ and $\lambda_2\in\mathbb{R}$, $F( \lambda_1 f + \lambda_2\underline{1})=\lambda_1 F(f)+\lambda_2 F(\underline{1})$) and $F(\underline{1})=1$. Dually, every functional $F:C(X)\rightarrow\mathbb{R}$ with these properties is of the form $ue_A$ for some convex closed nonempty set $A\in\mathcal{P}\mathcal{D}(X)$.
\end{theorem}
Note how Theorem \ref{riesz_representation} can be seen as specializing Theorem \ref{convex_representation} for singletons $A\!=\!\{\mu\}$. The following is another useful specialization of Theorem \ref{convex_representation}.
\begin{theorem}\label{sets_representation}
Let $X$ be a finite set, $A\in\mathcal{P}(X)$ a nonempty subset of $X$ and $B=\overline{\textnormal{H}}(\{\delta(x) \ | \ x\in A\})$ the convex hull of Dirac's over $A$, which is closed since $A$ is finite. The functional $ue_B$  preserve sups, i.e., $ue_B(f\sqcup g) \!=\! ue_B(f) \sqcup ue_B(g)$ where $f\sqcup g$ denotes the pointwise l.u.b. defined as $(f\sqcup g) (x) = f(x)\sqcup g(x)$. Dually, every $F$ as in Theorem \ref{convex_representation} which preserves  sups is of the form $ue_B$. 
\end{theorem} 
\begin{proof}
The first part of the Theorem is simple to verify. For the other direction, let $C\subseteq\mathcal{D}(X)$ be the convex closed set, given by Theorem \ref{convex_representation}, such that $F\!=\!ue_C$. Suppose towards a contradiction that $C$ is not of the form $B$. 
Let $Y\!=\!\{Êx \ | \ \exists \mu\!\in\! C. \mu(x)\!>\!0\}\!\subseteq\! X$ so that every probability distribution in $C$ is a convex combination of Dirac's $\delta_x$ for $x\in Y$. For each $x\!\in\! Y$ denote with $1_x\!:\!X\rightarrow\!\{0,1\}$ the indicator function of $x$, so that $E_\mu(1_x)\!=\!\mu(x)$.
Since $C$ is not of the form $B$, there exists some $x\!\in\! Y$ such that $ue_C(1_x)\!<\!1$. Let $\mu\!\in\! C$ maximizing $1_x$ in $C$, i.e., such that $ue_C(1_x)\!=\!E_\mu(1_x)$. The existence of $\mu$ follows from the compactness of $C$. Let $\lambda\!=\! \mu(x)$. Pick some other $y\!\in\! Y$ such that $\mu(y)\!>\!0$ and define $g\!=\!\lambda 1_y$.
It follows that $E_\mu(1_x \sqcup g) \!>\! \lambda$ and this implies $ue_C(1_x\sqcup g)\!>\!\lambda$. But, by construction, $ue_C(1_x)=\lambda$ and $ue_C(g)\leq \lambda$. A contradiction with the assumption that $F\!=\!ue_C$ preserves sups.
\end{proof}

Note that these representation results deal with nonempty sets. For our purposes it is useful to generalize them to include also the representation of $ue_\emptyset$ $(f\mapsto 0)$. Note that all functionals $F$ considered so far are monotone and positively affinely homogeneous (p.a.h.).
\begin{lemma}
Let $X$ be a finite set and $F\!:\!C(X)\!\rightarrow\! \mathbb{R}$ be p.a.h. and monotone. If $F(\underline{1})\!=\!0$ then $F(f)\!=\!0$, for all $f\!\in\! C(X)$.
\end{lemma}
\begin{proof}
Find $\lambda\!\geq\! 0$ such that $-\lambda \underline{1}\sqsubseteq f \sqsubseteq \lambda \underline{1}$. By positive homogeneity, $F(-\lambda \underline{1})=F(\lambda\underline{1})=0$. Apply monotonicity.
\end{proof}
Thus, by relaxing the requirement $F(\underline{1})\!=\!1$ in the three representation theorems to  $F(\underline{1})\!\in\!\{0,1\}$, we obtain functionals which can also represent the empty set.

\subsection{The Riesz Modal Logic}\label{logic_R}
In this section we introduce a novel $\mathbb{R}$-valued modal logic which we name $\R$ in honor of Frigyes Riesz who introduced what today are known as \emph{lattice ordered vector spaces} or \emph{Riesz spaces} \cite{Riesz1928}. We refer to \cite{JVR1977,Luxemburgh} as standard references to the theory of Riesz spaces. Here we limit ourselves to the basic definitions and results.
\begin{definition}
A \emph{Riesz space} is a real (or rational) vector space $R$ endowed with an order relation $\sqsubseteq$ which is a lattice and compatible with the vector space structure, i.e., if $f,g\in R$ then $f\sqsubseteq g$ implies  $f+h\sqsubseteq g+h$ for all $h\in R$ and $f\sqsupseteq \underline{0}$ implies $\lambda f\sqsupseteq \underline{0}$, for all $\lambda\geq 0$, where $\underline{0}$ denotes the zero element of the vector space $R$. As usual $f\sqcup g$ denotes the least upper bound (lub) of $f$ and $g$ and similarly for $f \sqcap g$ (glb). It is always the case that $f\sqcap g = -(-f\sqcup -g)$. Note that $f\sqsubseteq g$ if and only if $f\sqcup g=g$.
\end{definition}
The vector space $X\!\rightarrow\!\mathbb{R}$, with functions ordered pointwise, is an example of Riesz space. The language of Riesz spaces, combining linear algebra and order theory, is sufficiently rich to express the representation theorems of Section \ref{sec_representation_theorems}. 
\begin{definition}
Let $R$ be a Riesz space. The \emph{positive part} of $f\!\in\! R$ is defined as $f^+\!=\!f\sqcup \underline{0}$.
We denote with $R^+$ the set $\{ f^+ \ | \ f\in R\}$.
We say $R$ is \emph{Archimedean} if, for all $f\!\in\! R^+$, the greatest lower bound of $\{ (\frac{1}{n+1}) f \ | \ n\in\mathbb{N} \}$ exists and equals $\underline{0}$. An element $f\!\in\! R^+$ is a \emph{strong unit} if for all $g\!\in\! R^+$ there exists some $n\in\mathbb{N}$ such that $n f \sqsupseteq g$. We say $R$ is \emph{unitary} if it contains a strong unit. 
\end{definition}
Note that $X\!\rightarrow\!\mathbb{R}$ is Archimedean and unitary with the constant function $\underline{1}$ $(x\mapsto 1)$ as  strong unit. 

\begin{definition}[Boolean Elements]\label{boolean_elements}
Let $R$ be a Riesz space with strong unit $u$. Define the binary operation $\oplus$ of \emph{truncated sum} as $f\oplus g=(f+g)\sqcap u$. An element $f\in R$ is \emph{Boolean} if $f\oplus f=f$. Denote with $B$ the set of Boolean elements in $R$. It is known (see, e.g., \cite{MundiciBook}) that $(B,\sqcup, \neg, \underline{u})$, with $\neg b= u-b$ is a Boolean algebra. If $R$ is of the form $X\rightarrow\mathbb{R}$, with the pointwise order and strong unit $\underline{1}$, then $B$ is the Boolean algebra of functions $X\rightarrow\{0,1\}$.
\end{definition}

The following result, known as (finite dimensional) Yosida representation theorem,  is fundamental in the theory of Riesz spaces. 
\begin{theorem}[p. 152, \cite{Luxemburgh}]\label{yosida_theorem}
Every Riesz space $(R,\sqsubseteq)$ which is $n$-dimensional as a vector space, Archimedean and unitary is of the form $X\rightarrow\mathbb{R}$ with $|X|=n$, pointwise ordering and strong unit $\underline{1}$.
\end{theorem}

We now use this result, together with the representation theorems of Section \ref{sec_representation_theorems}, to develop a uniform algebraic account of PNTS's, Markov processes and NTS's (cf. Proposition \ref{NTS_as_PNTS}). The following definition is similar to that of Boolean algebras with operators, the algebraic counterpart of standard modal logic \cite{BdRVModal}.

\begin{definition}\label{definition_modal_riesz_space}
A \emph{modal Riesz space}  is a triple  $(R,\sqsubseteq,\Diamond)$ where $(R,\sqsubseteq)$ is a Riesz space  with strong unit (denoted by $\underline{1}$) and $\Diamond$ is an unary operation on $R$. We say that $(R,\sqsubseteq,\Diamond)$ is of \emph{type-PNTS} if: 
\begin{enumerate}
\item[] $\! \! \!\!\! $(Monotone) if $f\sqsubseteq g$ then $\Diamond f \sqsubseteq \Diamond g$
\item[] $\! \! \!\!\! $(Sublinear) $\Diamond(f+g)\sqsubseteq \Diamond(f) + \Diamond (g)$
\item[] $\! \! \!\!\! $(Pos. Affine Homogenous) $\Diamond(\lambda_1 f + \lambda_2 \underline{1})= \lambda_1\Diamond(
 f) + \lambda_2\Diamond( \underline{1})$\\
 for all $\lambda_1\geq 0$ and $\lambda_2\in\mathbb{R}$
\item[] $\! \! \!\!\! $($\Diamond\underline{1}$ is Boolean)  $\ \Diamond\underline{1}\oplus\Diamond\underline{1}=\Diamond\underline{1}$
\end{enumerate}
It is of \emph{type-MP} if it is of type-PNTS and further satisfies
\begin{enumerate}
\item[] $\! \! \!\!\! $(Linear) $\Diamond(\lambda_1 f+ \lambda_2 g)=\lambda_1 \Diamond(f) + \lambda_2 \Diamond(g)$,  for all $\lambda_1,\lambda_2\in\!\mathbb{R}$
\end{enumerate}
and of type-NTS  if it is of \emph{type-PNTS} and further satisfies
\begin{enumerate}
\item[] $\! \! \!\!\! $(Sup-preserving) $\Diamond(f\sqcup g)=\Diamond(f) \sqcup \Diamond(g)$.
\end{enumerate}
\end{definition}

The  following results should now be expected.

\begin{proposition}\label{model_to_algebra}
Let $X$ be a finite set and $(X,\alpha)$ a PNTS such that $\alpha(x)$ is convex closed, for all $x\in X$. Let $R\!=\!(C(X),\sqsubseteq,\Diamond_\alpha)$, with $\sqsubseteq$ the pointwise order and $\Diamond_\alpha$ as in Equation \ref{definition_diamond}. Then $R$ is a finitely dimensional Archimedean 
modal Riesz space of type-PNTS. If $(X,\alpha)$ is a Markov process then $R$ is of type-MP. If $(X,\alpha)$ is a NTS then  $R$ is of type-NTS.
\end{proposition}

\begin{theorem}\label{algebra_to_model}
Let $(R,\sqsubseteq,\Diamond)$ be a finite dimensional Archimedean modal Riesz space of type-PNTS. Then there exists a unique (up to isomorphism) PNTS $(X,\alpha)$, with $\alpha(x)$ convex closed for all $x\in X$, such that $(R,\sqsubseteq,\Diamond)$ is of the form $(C(X),\sqsubseteq,\Diamond_\alpha)$. If $R$ is of type-MP, then $(X,\alpha)$ is a Markov process. If $R$ is of type-NTS, then $(X,\alpha)$ is a NTS.
\end{theorem}
\begin{proof}
By Theorem \ref{yosida_theorem} $R$ is of the form $(X\rightarrow\mathbb{R})$ for some finite $X$. Of course the set $X$ is unique up to isomorphism. Thus the $\Diamond$ operation has type $(X\rightarrow\mathbb{R})\rightarrow (X\rightarrow\mathbb{R})$. By currying, consider the equivalent map $\diamond:X\rightarrow(X\rightarrow\mathbb{R})\rightarrow\mathbb{R}$. For every $x\in X$, the map $\diamond(x)$ is monotone, subadditive, p.a.h. and $\diamond(\underline{1})(x)\in\{1,0\}$. By Theorem \ref{convex_representation} there exists a unique convex closed set $C_x$ represented by $\diamond(x)$. The desired PNTS is then $(X,\alpha)$ with $\alpha(x)=C_x$. The cases for MP and NTS's are dealt in the same way by using Theorem \ref{convex_representation} and Theorem \ref{sets_representation}.
\end{proof}

The Riesz modal logic $\R$ which we now introduce is a syntactic counterpart of modal Riesz spaces.
\begin{definition}[Logic \R]
The formulas of the logic $\R$ are generated by the grammar: $$\phi::= \underline{1} \ | \ \phi + \phi \ \ | \  \ q\phi \ | \ \phi\sqcup \phi  \ | \ \Diamond \phi$$ where $q\!\in\!\mathbb{Q}$. We use the expected derived formulas $-\phi = (-1)\phi$, $\underline{0}\!=\! \underline{1} -\underline{1} $, $\phi\sqcap \psi\! =\! -(-\phi \sqcup -\psi)$, $\phi\oplus \psi\!=\!(\phi+\psi)\sqcap\underline{1}$, $\phi^+\!=\!\phi\sqcup\underline{0}$. We write $\phi\!\in\!\R$ if  $\phi$ is a formula of the logic $\R$. 
\end{definition}

\begin{definition}[Model $\mathbb{R}$-valued Semantics]\label{model_semantics}
Given a PNTS $(X,\alpha)$, we define the semantics $\sem{\phi}_\alpha\! :\!X\!\rightarrow\!\mathbb{R}$ of $\phi\in\R$  as: 
\begin{itemize}
\item $ \sem{\underline{1}}_{\alpha}(x)\!=\!1 $
\item  $\sem{\phi+\psi}_\alpha(x)\!=\!\sem{\phi}_\alpha(x)+\sem{\psi}_\alpha(x)$
\item $\sem{\phi\sqcup\psi}_\alpha(x)\!=\!\max\{ \sem{\phi}_\alpha(x), \sem{\psi}_\alpha(x)\}$
\item $\sem{q\phi}_\alpha(x)\!=\!q\sem{\phi}_\alpha(x)$
\item $\sem{\Diamond\phi}_{\alpha}\!=\!\Diamond_{\alpha}(\sem{\phi}_{\alpha})$, with $\Diamond_{\alpha}$ as in Equation \ref{definition_diamond}.
\end{itemize}
\end{definition}

\begin{proposition}[Soundness]\label{soundness_E}
Let $(X,\alpha)$ be a PNTS. For every formula $\phi$ and UE-bisimulation $E$ the map $\sem{\phi}_\alpha$ is $E$-invariant.
\end{proposition}
\begin{proof}
The proof goes by induction on the structure of $\phi$. The interesting case is $\sem{\Diamond \psi}_\alpha$. By definition of UE-bisimilarity, $(x,y)\in E$ holds if for all $E$-invariant $f:X\rightarrow\mathbb{R}$, the equality $\bigsqcup_{x\rightarrow \mu}E_\mu(f)= \bigsqcup_{y\rightarrow \nu}E_\nu(f)$ holds. By induction hypothesis $\sem{\psi}_\alpha$ is $E$-invariant. Apply Definition \ref{model_semantics}.
\end{proof}
\begin{remark}
Note that Proposition \ref{soundness_E} can be rephrased by saying that, given a PNTS $(X,\alpha)$, every UE-bisimulation $E$ is contained in the equivalence relation $\textnormal{ker}(\R)$ defined as follows:  $\textnormal{ker}(\R)= \{ (x,y)\ | \ \forall \phi. \big( \sem{\phi}_\alpha(x) = \sem{\phi}_\alpha(y)\big)\}$. 
\end{remark}
We now prove that $\textnormal{ker}(\R)$ is itself a UE-bisimulation and thus, by Proposition \ref{soundness_E}, the greatest one. The proof will also show that $\R$ is strongly complete in the sense of Definition \ref{strong_completeness_def}. In the proof we make use of the following important result which, for the sake of simplicity, we specialize to finite sets (see, e.g., \cite[\S 13]{JVR1977}).
\begin{proposition}[Stone-Weierstrass Theorem for Riesz spaces]\label{stone_weierstrass}
Let $X$ be a finite set and $R$ a Riesz subspace of $X\rightarrow\mathbb{R}$ that \emph{separates points}, i.e., such that for every $x\neq y\in X$ there exists $f\in R$ such that $f(x)\neq f(y)$. Then $R$ is dense in $X\rightarrow\mathbb{R}$.
\end{proposition}
\begin{theorem}\label{completeness_th_1}
Let $(X,\alpha)$ be a PNTS. Then $\textnormal{ker}(\R)$ is a UE-bisimulation. 
\begin{proof}
Denote with $E$ the equivalence relation $\textnormal{ker}(\R)$. Note that the map $\sem{\phi}_\alpha$, for every $\phi\!\in\!\R$, can be seen as a map of type $C(X/E)= X/E\rightarrow\mathbb{R}$ since, by definition, $\sem{\phi}_\alpha$ is $E$-invariant. For every $x,y\!\in\! X$ such that $(x,y)\!\not\in\! E$ there exists some $\phi\!\in\! \R$ such that $\sem{\phi}_\alpha(x) \!\neq\! \sem{\phi}_\alpha(y)$. 
Thus $\R$ separates points in $X/E$. It is clear that $\R$ carries a Riesz space structure and that $\{ \sem{\phi}_\alpha \ | \ \phi\in \R\}$ is a Riesz subspace of $C(X/E)$. Thus, by Theorem \ref{stone_weierstrass}, the (denotations of) formulas of $\R$ are dense $C(X/E)$.

We need to prove that if $(x,y)\!\in\! E$
then, for every $E$-invariant $f\!:\!X\rightarrow\mathbb{R}$, the equality $ \bigsqcup_{x\rightarrow \mu}E_{\mu}(f)\!=\!\bigsqcup_{y\rightarrow \nu}E_{\nu}(f)$ holds. This equality can equivalently be expressed as $\big(\Diamond_{\alpha}(f)\big)(x) \!=\! \big(\Diamond_{\alpha}(f)\big)(y)$. Assume towards a contradiction that 
$\big(\Diamond_{\alpha}(f)\big)(x)\!>\!\big(\Diamond_{\alpha}(f)\big)(y)$. Since $f$ is $E$-invariant, it can be seen as a map of type $X/E\!\rightarrow\!\mathbb{R}$. Since $\R$ is dense in this space, there exists $\phi\in\R$ (approximating $f$) such that $\big(\Diamond_{\alpha}(\phi)\big)(x)\!>\!\big(\Diamond_{\alpha}(\phi)\big)(y)$. By definition, this means that $\sem{\Diamond \phi}_\alpha(x)\!>\! \sem{\Diamond \phi}_\alpha(y)$. A contradiction with the hypothesis that $(x,y)\in \textnormal{ker}(\R)$.
\end{proof}
\end{theorem}
\begin{corollary}\label{completeness_th_2}
The logic $\R$ is strongly complete.
\end{corollary}

We considered an $\mathbb{R}$-valued semantics for $\R$ but, in light of Proposition \ref{model_to_algebra} and Theorem \ref{algebra_to_model}, an equivalent algebraic semantics can be defined.
\begin{definition}
Let $(R,\sqsubseteq,\Diamond)$ be a modal Riesz space  with strong unit denoted by $u$. Denote by $\gsem{\_ }_R\!:\!\R\!\rightarrow\! R$ the expected homomorphism, from formulas to $R$,  specified by $\gsem{\underline{1}}=u$. The element $\gsem{\phi}_R$ is called the \emph{algebraic semantics} of $\phi$ interpreted in $R$.
\end{definition}

\begin{definition}[Equational Theory]
For a PNTS $(X,\alpha)$, we write $(X,\alpha)\models \phi=\psi$ if $\sem{\phi}_\alpha = \sem{\psi}_\alpha$. 
For a collection $\mathcal{M}$ of PNTS's we write $\mathcal{M}\models \phi=\psi$ if $(X,\alpha)\models \phi=\psi$ for all $(X,\alpha)\in \mathcal{M}$.
Similarly, for a modal Riesz space $R$, we write $R\Vdash \phi=\psi$ if $\gsem{\phi}_R = \gsem{\psi}_R$ and, for a collection $\mathcal{R}$ of modal Riesz space, we write $\mathcal{M}\models \phi=\psi$ if $R\Vdash \phi=\psi$ for all $R\in \mathcal{R}$.
\end{definition}
The following is a direct consequence of Proposition \ref{model_to_algebra} and Theorem \ref{algebra_to_model}.
\begin{corollary}\label{equational_equivalence}
Let $\mathcal{M}$ be the collection of all PNTS (resp. Markov processes, NTS's) $(X,\alpha)$ with $X$ finite and $\alpha(x)$ closed. Let $\mathcal{R}$ be the collection of finite dimensional Archimedean modal Riesz space of type-PNTS (resp. type-MP, type-NTS). Then for all $\phi,\psi\in \R$ it holds that $\mathcal{M}\models \phi\!=\!\psi$ if and only if  $\mathcal{R}\Vdash \phi\!=\!\psi$.
\end{corollary}

As we anticipated in the introduction, we are not aware of any other $\mathbb{R}$-valued or Boolean (i.e., $\{0,1\}$-valued) logic for PNTS's equipped with an algebraic (i.e., axiomatic) semantics. In the context of Markov processes, on the other hand, the Boolean logic of Panangaden \emph{et.\ al.\ } \cite{StonePrakash} is completely axiomatized and a duality  between (generally infinite) models  and algebras is proved. It is going to be interesting to study the relation between their logic and $\R$. With respect to other logics for Markov processes, we mention that the $\mathbb{R}$-valued (in fact $[0,1]$-valued) logic of \cite{PrakashBook} can be easily seen as a fragment of $\R$.


Note that $\mathbb{R}$-valued logics are interesting also in the context of NTS's, despite the fact that ordinary Boolean ones (such as the modal logic $\K$) suffice to characterize Milner and Park's bisimilarity. See, e.g., \cite{FGK2010}, for a (fixed-point) $\mathbb{R}$-valued modal logic, for expressing quantitative properties of ordinary NTS's. 
Thus, in light of its algebraic presentation and completeness result, the logic $\R$ is also useful as a foundation for quantitative logics for NTS's. Note that the standard modal logic $\K$ can be seen as a Boolean fragment of $\R$ (cf. Definition \ref{boolean_elements}). The logic $\K$ can be axiomatized by the axioms of Boolean algebra plus $\phi\sqsupseteq \psi\Rightarrow\Diamond\phi\sqsupseteq \Diamond\psi$ (monotonicity), $\Diamond \underline{0}=\underline{0}$ (which follows easily since $\Diamond$ is a p.a.h. functional) and $\Diamond(\phi \sqcup \psi)=\Diamond(\phi)\sqcup \Diamond(\psi)$ (cf. Axiom of sup-preservation) \cite{BdRVModal}. We plan to study the relations between $\K$ and $\R$ in future works.


The correspondence between model and algebraic semantics, defined on three important classes of systems, 
is likely to be a valuable tool for future investigations.
 As an example, we suggest that
the algebraic axiomatization could be of help towards design of \emph{structural proof systems} (e.g., \emph{sequent calculi}) for the logic $\R$. Such systems proved to be  useful in the development of techniques for \emph{compositional verification} of nondeterministic and probabilisitic systems. See, e.g., \cite{Simpson04} for a paradigmatic (but restricted to NTS's) example on this line of research based on the standard modal logic $\K$ and, e.g., \cite{Mio2012c,CLM2011}, for recent developments.

\begin{remark}\label{remark_logic_terminology}
We freely adopted the term \emph{logic} for the real valued formalism of $\R$. This is well established terminology in computer-science literature (see, e.g., \cite{MM07,HM96,FGK2010}) as well as in  mathematical logic. For example, \L ukasiewicz logic \cite{MundiciBook} and, more generally, \emph{fuzzy logics}, are well known examples of real-valued formalisms. Furthermore, if a logic is understood as a language for expressing properties of models,  then $\R$ is rightfully a logic in that its formulas denote \emph{experiments} to be performed on models, and these can be regarded as \emph{quantitative predicates}. 
\end{remark}

\subsection{Probabilistic Modal $\mu$-Calculi}\label{logic_mu_calculi}
The choice of primitives of the logic $\R$ has been guided by the Representation theorems of Section \ref{sec_representation_theorems}. The language of Riesz spaces is a simple, yet mathematically rich \cite{JVR1977,Luxemburgh,Riesz1928}, formalism capable of formulating these powerful results from functional analysis. We see the value of the logic $\R$ mainly as a tool for theoretical investigations. However, as many other basic modal logics for NTS's including the modal logic $\K$, the logic $\R$ cannot express many interesting properties of systems such as: termination goals, safety and liveness constraints, \emph{etc}. Richer logics, capable of expressing these properties, can be obtained by enriching the basic modal systems with additional operators. The modal $\mu$-calculus of Dexter Kozen \cite{Kozen83} is a very expressive logic for NTS's obtained by extending $\K$ with least ($\mu$) and greatest ($\nu$) fixed point operators \cite{Stirling96,BS2001}. The modal $\mu$-calculus enjoys a remarkably rich theory \cite{Rudiments2001}. 

In the last decade several fixed point modal logics (henceforth, \emph{probabilistic $\mu$-calculi}) for expressing properties of PNTS's have been considered. We refer to \cite{MioThesis} as a reference to this area of research. Among the different approaches that have been followed to developing analogues of the modal $\mu$-calculus, the most relevant is based on enriching some basic quantitative $\mathbb{R}$-valued modal logic with least and greatest fixed point operators. 
\begin{definition}\label{modal_logics_def_general}
The formulas of the logics $\textnormal{qL}$, $\textnormal{qL}^\odot$, $\textnormal{qL}^\ominus$ and \L\ are defined as:
\begin{center}
\begin{tabular}{l l l l l l l l l }
$(\textnormal{qL})$ &  & $\phi::=  \underline{1}  \ |\ \phi \sqcup \phi \ | \ q\phi \ |\  \neg \phi \ | \ \Diamond \phi$\\
$(\textnormal{qL}^\odot)$ & & $\phi::=  \underline{1}  \ |\ \phi \sqcup \phi \ | \ q\phi \ | \ \neg \phi \ | \ \Diamond \phi\ | \ \phi \cdot \phi$ \\
$(\textnormal{qL}^\ominus)$ & & $\phi::=  \underline{1}  \ |\ \phi \sqcup \phi \ | \ q\phi \ | \ \neg \phi \ | \ \Diamond \phi \ | \ \phi \ominus q$\\
(\L) & $ \ \ \ $ & $\phi::=  \underline{1} \ |\  \phi \sqcup \phi \ | \ q\phi \ |\  \neg \phi \ | \ \Diamond \phi\ | \ \phi \oplus \phi$\\
\end{tabular}
\end{center}
where $q\!\in\!{[0,1]\cap\mathbb{Q}}$ is a rational number. The semantics of a formula $\phi$ interpreted in a PNTS $(X,\alpha)$ is defined as a map $\sem{\phi}_{\alpha}:X\rightarrow\mathbb{R}$ specified as (subscripts $\alpha$ are omitted): $\sem{\underline{1}}$, $\sem{\phi\sqcup\psi}$, $\sem{q\phi}$ and $\sem{\Diamond\phi}$  are defined as in $\R$; $\sem{\neg \phi}(x)=1-\sem{\phi}(x)$, $\sem{\phi\cdot \psi}(x)=\sem{\phi}(x)\sem{\psi}(x)$ (product on reals), $\sem{\phi\ominus q}(x)=\max\{Ê0, \sem{\phi}(x)-q\}$, and $\sem{\phi\oplus \psi}(x)=\min\{1, \sem{\phi}(x)+\sem{\psi}(x)\}$.
\end{definition}
Note that $\textnormal{qL}^\ominus$ is a fragment of $\L$ ($\phi\ominus q$ can be encoded as $\neg(\neg \phi \oplus q)$) and that $\L$ is a fragment of $\R$ ($\phi\oplus \psi$ can be encoded as $(\phi+\psi)\sqcap \underline{1}$). The operations of the logic $\L$ coincide with those available\footnote{\label{RMV_structure_of_L}Actually, the operation of scalar multiplication by $q\in[0,1]$, gives $\L$ a so-called \emph{Riesz MV-algebra} structure. See, e.g., \cite{RMV2011}.} in \L ukasiewicz logic \cite{MundiciBook} and this justifies the choice of notation. The logic $\R$ and $\textnormal{qL}^\odot$, on the other hand, appear to be incomparable. Furthermore, observe that $\sem{\phi}_\alpha$ ranges over $[0,1]$, i.e.,  has always type $\sem{\phi}_\alpha\!:\!X\!\rightarrow\![0,1]$, and that the interpretation of every connective (except $\neg$) is pointwise monotone. The latter two observations allow one to develop a theory of fixed points on top of these logics based on the order-theoretic Knaster-Tarski fixed-point theorem. The syntax of the logics is extended to include \emph{variables} $\texttt{Var}$, ranged over by the letter $v$, and least ($\mu v.\phi$) and greatest ($\nu v.\phi$) fixed point operators which, as expected, bind the variable $v$ in $\phi$. As customary in fixed-point logics \cite{Stirling96}, one restricts attention to \emph{positive formulas}, where
every occurrence of variable $v$ occurs in the scope of an even number of
negations. The semantics of formulas is defined as follows.
\begin{definition}
For a PNTS $(X,\alpha)$, an interpretation of the variables is an assignment $\rho\!:\!\texttt{Var}\!\rightarrow \!(X\!\rightarrow\![0,1])$. For every $f\!:\!X\rightarrow[0,1]$ we write $\rho[f/v]$ for the interpretation specified as $\rho[f/v](w)\!=\!f$ if $w\!=\!v$ and $\rho[f/v](w)\!=\!\rho(w)$ otherwise.  The semantics of a fixed point formula $\phi$ interpreted in $(X,\alpha)$ with $\rho$ is the map specified extending Definition \ref{modal_logics_def_general} with:   $\sem{v}_\rho\!=\! \rho(v)$, $\sem{\mu v.\phi}_\rho\!=\! \lfp( f\mapsto \sem{\phi}_{\rho[f/v]})$ and, similarly, $\sem{\nu v.\phi}_\rho\!=\! \gfp( f\mapsto \sem{\phi}_{\rho[f/v]})$, where $f$ ranges over $X\rightarrow[0,1]$.
\end{definition}
Since $X\rightarrow[0,1]$ (ordered pointwise) is a complete lattice, least and greatest fixed points of arbitrary monotone operators exists. We denote with $\textnormal{qL}\mu$, $\textnormal{qL}\mu^\odot$, $\textnormal{qL}\mu^\ominus$ and $\L\mu$ the $\mu$-calculi obtained in this way. Historically, qL$\mu$ has been the first logic of this family to be studied \cite{HM96,MM07}. 
The logic qL$\mu^\ominus$ has been considered in, e.g., \cite{AM04,deAlfaro2008} and is based on the quantitative modal logic (for Markov processes) of Panangaden \cite{PrakashBook}. It is possible to show (see, e.g., \cite{MIO2012b} that these logics are not sufficiently expressive to encode other important temporal logics for PNTS's such as the \emph{probabilistic CTL} (PCTL) of \cite{BA1995}. To address this limitation, the two logics qL$\mu^{\odot}$ and $\L \mu$ are introduced\footnote{The logic $\L\mu$ is called pL$\mu_{\oplus}$ in \cite{MioThesis}.} in \cite{MioThesis}. The author shows how qL$\mu^{\odot}$ and $\L\mu$ can encode the \emph{qualitative fragment} of PCTL and \emph{full} PCTL respectively.
Our first result of this section is the soundness, with respect to UE-bisimilarity, of all $\mu$-calculi discussed above.
\begin{proposition}[Soundness]
Let $(X,\alpha)$ be a PNTS, $E$ an UE-bisimulation and $\rho$ an interpretation such that $\rho(v)$ is $E$-invariant for all $v\in\texttt{Var}$.  If $(x,y)\!\in\! E$ then, for every closed formula $\phi$ of qL$\mu$ it holds that $\sem{\phi}_\rho(x)=\sem{\phi}_\rho(y)$.  Similarly for $qL\mu^\ominus$, qL$\mu^\odot$ and $\L\mu$ formulas.
\end{proposition}
\begin{proof}
The proof is by induction of the complexity of $\phi$ as in Proposition \ref{soundness_E}. For the case $\phi=\mu v.\psi$, by Knaster-Tarski theorem, we have $\sem{\phi}_\rho\!=\!\sem{\phi}^\alpha_{\rho}$, where $\alpha$ ranges over the ordinals and $\sem{\psi}^{\alpha+1}_{\rho}\!=\!\sem{\psi}_{\rho[\sem{\psi}^\alpha_\rho/v]}$ and $\sem{\psi}^\beta_\rho\!=\!\bigsqcup_{\alpha<\beta}\sem{\psi}^\beta_\rho$, where $\beta$ is a limit ordinal. It is clear that the pointwise supremum of a familiy $E$-invariant functions is $E$-invariant. Furthermore, by induction hypothesis on $\psi$ and $\alpha$, also $\sem{\psi}^{\alpha+1}_\rho$ is $E$-invariant. Thus $\sem{\phi}$ is $E$-invariant as desired. The case for $\phi=\nu v.\psi$ is similar.
\end{proof}

It is shown in \cite{deAlfaro2008} that ${\textnormal{qL}\mu}^\ominus$ (or even its fixed-point free fragment $\textnormal{qL}^\ominus$) is weakly complete with respect to UE-bisimilarity: if $\sem{\phi}(x)\!=\! \sem{\phi}(y)$ for all qL$\mu^\ominus$ formulas $\phi$, then  $x$ and $y$ are UE-bisimilar.
However the logic ${\textnormal{qL}\mu}^\ominus$ is not strongly complete (cf. Section \ref{metrics_section} below) in the sense\footnote{Actually, Definition \ref{strong_completeness_def} requires (the denotation of) formulas to be dense in $X\rightarrow\mathbb{R}$. In the context of $[0,1]$-value logics, we instead require density in $X\rightarrow[0,1]$. Note how this adapted notion of strong completeness still implies weak completeness.} of Definition \ref{strong_completeness_def}. We now prove that, instead, the logics qL$\mu^\odot$ and $\L$, and thus also their fixed-point extensions, are strongly complete.
\begin{theorem}\label{completeness_L}
The logics \L\ and $\textnormal{qL}^\odot$ are strongly complete.
\end{theorem}
\begin{proof}
The proof for $\L$ is based on a Stone-Weierstrass Theorem for RMV-algebras (see, e.g., \cite{RMV2011}), stating that every RMV-subalgebra (cf. Footnote \ref{RMV_structure_of_L}) of $X\rightarrow[0,1]$ that separates points is dense in $X\rightarrow[0,1]$. The proof is then identical to that of Theorem \ref{completeness_th_1}. For what concerns $\textnormal{qL}^\odot$ it is shown in \cite{MioThesis} (Lemma 3.3.16) that for for every $f,g\!\in\!X\rightarrow[0,1]$, the function $f\oplus g$ can be approximated by a sequence of functions expressible in the language of $\textnormal{qL}^\odot $. 
\end{proof}
Thus the expressive logics $\textnormal{qL}\mu^\odot$ and $\L\mu$ can be used to denote (approximations of) all possible $[0,1]$-valued experiments (cf. Example \ref{example_formulas_as_experiments}). This, together with the fact that $\L\mu$ can encode the logic PCTL, provides a strong characterization of UE-bisimilarity in terms of expressive, and potentially useful in practice, temporal logics for verification.

\subsection{Logical Characterization of the Hausdorff Metric}\label{metrics_section}
As mentioned above, the logic qL$\mu^\ominus$ of \cite{deAlfaro2008} (and the similar $[0,1]$-valued logic of Panangaden \cite{PrakashBook}) is not strongly complete but just weakly complete. The connectives of qL$\mu^\ominus$ are carefully chosen so that the denotation of formulas are Lipschitz (i.e., not expansive) functions (see  \cite{PrakashBook}). This property is crucial in proving that the logically defined \emph{behavioral metric}\footnote{Technically, $d(x,y)$ is a pseudo-metric because if $x,y$ are UE-bisimilar then $d(x,y)=0$, even if $x\neq y$. The function $d(x,y)$ is an authentic metric on the 
 state space quotiented by UE-bisimilarity (cf. Remark \ref{convex_algo}). To ignore this pedantic distinction, in this section we just simple assume that the largest UE-bisimulation in $(X,\alpha)$ is the identity relation.} on  states of a PNTS $(X,\alpha)$
\begin{equation}\label{def_metric_equation}
d_L(x,y)= \displaystyle \bigsqcup_{\phi\in \textnormal{qL}\mu^\ominus} | \sem{\phi}_\alpha(x) - \sem{\phi}_\alpha(y) |
\end{equation}
coincides with the so-called \emph{Hausdorff} behavioral metric which is based on the following idea. 
For a given PNTS $(X,\alpha)$, and for any\footnote{Many natural metrics on $\mathcal{D}(X)$ exist. E.g., any norm on $\mathbb{R}^n\!\cong\!X\!\rightarrow\!\mathbb{R}$ (cf. Convention \ref{conv_1}) induces a metric on $\mathcal{D}(X)$. The sup-norm on $\mathbb{R}^n$ induces the so-called \emph{total variation} metric on $\mathcal{D}(X)$. See also Remark \ref{remark_unique_topology}.} metric $m$ on $\mathcal{D}(X)$,
one can consider the Hausdorff metric $d^m_H$ on the space of closed subsets of $\mathcal{D}(X)$ (see, e.g., \cite{Kechris}). Then it is natural to define a metric on states as $d(x,y)=d^m_H(\alpha(x),\alpha(y))$. When $m$ is the Kantorovich metric (see, e.g., \cite{Kechris}) on $\mathcal{D}(X)$, then $d^m_H$ and $d_L$, as above, coincide. This, as claimed earlier, implies that $\textnormal{qL}\mu^\ominus$ is not strongly complete: take $x,y$ such that $0\!<\!d(x,y)\!<\!1$ so that there is no $\phi\!\in\!\textnormal{qL}\mu^\ominus$ such that $\sem{\phi}(x)\!=\!1$ and $\sem{\phi}(y)\!=\!0$.

The introduction of behavioral metrics for probabilistic systems is strongly motivated by the need of \emph{approximation methods} (based on the idea that small changes of the probabilities in a PNTS corresponds to small changes in behavior), a line of research pioneered by Panangaden \cite{PrakashBook}.
The fact that qL$\mu^\ominus$ characterizes the Hausdorff behavioral metric is thus a remarkable property. Given the strong completeness of $\L$ (or $\R$, $\textnormal{qL}^\odot$), no interesting (i.e., not discrete) metric can be obtained as in (\ref{def_metric_equation}). We now show, however, that also the richer logic $\L$ (and $\R$, $\textnormal{qL}^\odot$) can be used to logically characterize the \emph{Hausdorff} behavioral metric in an interesting fashion. This is once again proven by applying results from linear algebra.
\begin{theorem}\label{metric_characterization}
Let $(X,\alpha)$ be a PNTS. Define $d_L$ as: $$d_L(x,y)=\displaystyle \bigsqcup_{\phi\in \textnormal{\L}} | \sem{\Diamond\phi}_\alpha(x) - \sem{\Diamond\phi}_\alpha(y) | $$
Then $d_L$ is (equivalent to) $d(x,y)\!=\!d^m_H(\alpha(x),\alpha(y))$, where $m$ is the total variation metric.
\end{theorem}
\begin{proof}
Assume without loss of generality (by Theorem \ref{theorem_UE=convex}) that $\alpha(x)$ is convex closed for all $x\!\in\! X$. It is known (see, e.g., \cite[\S 8.4]{LAX}) that for convex closed sets the equality $d^m_H(A,B)=\bigsqcup_{\| f\|_\infty \leq1} | ue_A(f)-ue_B(f)|$ holds, where $\| f\|_\infty\!\leq\!1$ means that $f$ has sup-norm $\leq\! 1$, i.e., $f\!:\!X\!\rightarrow\![-1,1]$. If $A\!=\!\emptyset$ then  $d_L(x,y)\!=\!d(x,y)$, as it is simple to verify.
Thus assume $A,B\!\neq\!\emptyset$. Since $ue_A$ and $ue_B$ are positive affinely homogeneous, for every $\| f\| \leq 1$ we have, for $g \!=\!  \frac{1}{2}f+\frac{1}{2}$, that $ue_A(g)\!=\! \frac{1}{2}ue_A(f)+\frac{1}{2}$ and similarly for $ue_B(g)$. Note that $g\!:\!X\rightarrow\![0,1]$. Thus $| ue_A(g)-ue_B(g)| \!= \!\frac{1}{2} | ue_A(f)-ue_B(f)|$. 
Hence the metric $d^\prime(A,B)\!=\!\bigsqcup_{g:X\rightarrow[0,1]} | ue_A(g)-ue_B(g)|$ is equivalent to $d^m_H(A,B)$, i.e., the two metrics induce the same topology. The desired result follows by density of $\L$ in $X\rightarrow[0,1]$ and the fact that $ue_{\alpha(x)}(\sem{\phi})\!=\!\sem{\Diamond\phi}_\alpha (x)$.
\end{proof}

\begin{remark}
Due to the lack of space, we just remark here that, given a finite PNTS $(X,\alpha)$, it is simple to calculate the value $d_L(x,y)$ as in Theorem \ref{metric_characterization} by means of Linear Programming.
\end{remark}

It is natural to interpret $d_L(x,y)$ as a value related to the probability of distinguishing between $x$ and $y$ in a \emph{one-shot} experiment $\phi$, modeled by the formula $\Diamond\phi$ (``push the button'' once and perform experiment $\phi$). We plan to investigate this viewpoint, and the possible relations with the (information-theoretic) one-shot attack models of \cite{KPP2008}, in future research.

\section{Congruence for PGSOS}\label{congruence_section}
A \emph{congruence} on an algebra $\mathcal{A}$ is an equivalence relation that respects the constructors of the algebraic structure: if $a_{i}\equiv b_{i}$, for $i\!\in\!\{1,\dots,n\}$, then $f(\vec{a})\equiv f(\vec{b})$ for any operation $f$ of arity $n$ in the signature of $\mathcal{A}$. When $\mathcal{A}$ is a collection of programs and each $f$ is a program constructor, a congruence relation captures a notion of behavioral equivalence and validates the principle that substituting behaviorally equal sub-programs $a_{i}$ of a compound program $f(\vec{a})$ with behaviorally equivalent sub-programs $b_{i}$, results in a program $f(\vec{b})$ that is equivalent to $f(\vec{a})$. We show in this section that UE-bisimilarity is a congruence with respect to the important program constructor of (communicating) parallel composition. The proof is valid for arbitrary systems (i.e., potentially infinite PNTS's) and can readily be adapted to show that UE-bisimilarity is a congruence with respect to the wide family of program constructors specified following the PGSOS rule format of \cite{BartelsThesis} which virtually includes all CCS-style operators of practical interest. To make our discussion interesting, we now consider \emph{labeled} PNTS's, i.e., structures $(X,\{\alpha_a\}_{a\in L})$ such that $(X,\alpha_a)$ is a PNTS, for every label $a\!\in\! L$.
We define UE-bisimulations on labeled PNTS's as expected.
\begin{definition}\label{UE_bisim_labeled}
Given a labeled PNTS $(X,\{\alpha_a\}_{a\in L})$, an equivalence relation $E\subseteq X\times X$ is a UE-bisimulation if it is a UE-bisimulation for the PNTS $(X,\alpha_a)$, for all $a\in L$.
\end{definition}

The operational semantics of the parallel operator operator ($||$) is specified as expected where, as customary, we consider a set of labels $L$ containing a distinct label $\tau$ and pairs a complementary labels $a,\overline{a}$ (with $\overline{\overline{a}}=a$). 
\begin{center}
\footnotesize{
\AxiomC{$x\freccia{a}\mu$}
\RightLabel{($||$ L) $\ \  \   $}
\UnaryInfC{$x|| y \freccia{a} \mu || y $}
\DisplayProof
\AxiomC{$y\freccia{a}\mu$}
\RightLabel{($||$ R) $ \ \ \ $}
\UnaryInfC{$x|| y \freccia{a} x || \mu $}
\DisplayProof
\AxiomC{$y\freccia{a}\mu$}
\AxiomC{$y\freccia{\overline{a}} \nu $}
\RightLabel{($||$ Comm)}
\BinaryInfC{$x|| y \freccia{\tau} \mu || \nu $}
\DisplayProof
}
\end{center}
Following standard approaches (see, e.g., \cite{Simpson04}) this leads to the corresponding notion of process algebra for $||$. 
\begin{definition}
A \emph{process algebra} for the parallel operator $||$ is a labeled PNTS $(X,\{\alpha_{a}\}_{a\in L})$ together with an interpretation $\iota$ of $||$, i.e., a function $\iota\!:\!X\times X\!\rightarrow\! X$, such that $\iota(x,y)\freccia{a}\mu$ holds iff:
\begin{enumerate}
\item ($||$ L): $x \freccia{a} \mu_{1}$ and $\mu= i(\mu_{1} \times \delta_{y}) $, or
\item ($||$ R): $y \freccia{a} \mu_{1}$ and $\mu= i(\delta_{x} \times \mu_{1}) $, or
\item ($||$ Comm): $a\!=\!\tau$, $x \freccia{b}\mu_{1}$, $y \freccia{\overline{b}} \mu_{2}$ and  $\mu= i(\mu_1 \times \mu_2)$
\end{enumerate}
hold, where $x\freccia{a}\mu$ iff $\mu\!\in\!\alpha_a(x)$, $\mu\!\times\! \nu$ denotes the product probability distribution ($(x,y)\mapsto \mu(x)\nu(y)$) and  $i(\mu \times \nu)$ is defined as 
$i(\mu \times \nu)(x) = \sum_{z \in \iota^{-1}(x)}(\mu \times \nu)(z)$, with $z\in X\times X$.
\end{definition}

\begin{theorem}
For a process algebra $\langle (X,\{\alpha_a\}_{a\in L}),\iota\rangle$,  the UE-bisimilarity relation is a congruence for $||$.
\end{theorem}
\begin{proof}
Let $E$ be the relation of UE-bisimilarity, i.e., the greatest UE-bisimulation on $X$ (cf. Theorem \ref{completeness_th_1}). 
In what follows we just write $x||y$ in place of $i(x,y)$. Define the relation $R\!=\!\big\{ (x  || y, x^{\prime}|| y), (y  || x, y || x^{\prime}) \ | \ (x,x^{\prime})\!\in\! E,\ y\!\in\!X \big\}$. Since $E$ is an equivalence relations so is $R$.  We prove the theorem by showing that $R$ is a UE-bisimulation and, as such, contained in $E$. We need to prove that for every $(z,z^{\prime})\!\in\! R$ and every $R$-invariant function $f\!:\!X\!\rightarrow\! \mathbb{R}$ the following two implications hold:
\begin{enumerate}[i]
\item if $z \freccia{a} \mu$ then for all $\epsilon >0$, there is $z^{\prime} \freccia{a}\nu$ such that $E_\mu(f) \leq E_\nu(f) + \epsilon$, 
\item if $z^{\prime} \freccia{a} \nu$ then for all $\epsilon >0$ there is $z \freccia{a}\mu$ such that  $E_\nu(f) \leq E_\mu(f) + \epsilon$, 
\end{enumerate}
for all $a\in L$. Note how we reformulated  $ue_{\alpha_{a}(z)}(f)\! =\! ue_{\alpha_{a}(z^\prime)}(f)$ as the equivalent conjunction of $(i)$ and $(ii)$.

We just consider the case $z\!=\!x||y$ and $z^{\prime}\!=\!x^{\prime}|| y$, with $(x,x^{\prime})\!\in\! E$, and show ($i$) above. The other cases are handled in a similar way. If $a\neq \tau$, we need to distinguish two cases:
\vspace{-1mm}
\begin{enumerate}
\item ($||$ R):  $y\freccia{a}\mu_1$ and $\mu= x || \mu_1$, and
\item ($||$ L): $x\freccia{a}\mu_1$ and $\mu= \mu_1 || y$
\end{enumerate}
\vspace{-1mm}
where we simply denoted with $x|| \mu_{1}$ and $\mu_{1}|| y$ the probability measures $\iota(\delta_{x}\times \mu_{1})$ and $\iota(\mu_{1}\times \delta_{y})$, respectively.

In the case  ($||$ R), choose $\nu$ as $\nu= x^{\prime}|| \mu_1$. It is then the case that $\mu( x|| y^{\prime})= \nu( x^{\prime}|| y^{\prime})$, for all $y^{\prime}\!\in\! X$. Since $(x,x^\prime)\!\in\! E$ and $f$ is $E$-invariant, the equality $f(x||y)\!=\!f(x^\prime||y)$ holds, for all $y\!\in\!X$. It then  follows that $E_\mu(f) = E_\nu (f)$.

 Consider now the case ($||$ L). Define $g_{y}\!:\!X\!\rightarrow\![0,1]$ as $g_{y}(z)\!=\!f(z||y)$. Note that $g_y$ is $E$-invariant because $f$ is $R$-invariant.  Since $(x,x^{\prime})\!\in\!E$, there exists some $x^{\prime}\freccia{a}\nu_1$ such that $E_{\mu_1} (g_{y})\! \leq\! E_{\nu_1}( g_{y})+\epsilon$. Take $\nu=\nu_1 \!\times\! \delta_y$ so that $x^\prime||y\rightarrow\nu$ holds. It is straightforward to verify that $E_{\mu_1} (g_y) \!= \! E_\mu( f)$ and, similarly, $E_{\nu_1} (g_y) \!= \! E_\nu( f)$. Thus the desired result follows.

If $a=\tau$, we need to consider, in addition to the cases ($||$ R) and  ($||$ L), the third case triggered by the rule ($||$ Comm). Thus assume 
\vspace{-1mm}
\begin{enumerate}
\item[3.] ($||$ Comm):  $x\freccia{a}\mu_{1}$, $y\freccia{\overline{a}}\mu_{2}$ and $\mu= \mu_{1} || \mu_{2}$.
\end{enumerate}
\vspace{-1mm}
Define $g\!:\!X\rightarrow[0,1]$ as $g(x)\!=\!\sum_y \mu_2(y)f(x||y)$. Note, as above, that $g$ is $E$-invariant because $f$ is $R$-invariant. Since $(x,x^{\prime})\!\in\!E$, there exists some $x^{\prime}\freccia{a}\nu_{1}$ such that $E_{\mu_1}( g) \!\leq\! E_{\nu_1}(g)+\epsilon$. As $\nu$ take the probability measure $\nu\!=\!\nu_{1} || \mu_{2}$. From definitions we have $E_{\mu_1}(g)\! = \!E_{\mu}( f)$ and $E_{\nu_1} (g) \!=\!E_\nu( f)$, as desired.
\end{proof}

\section{Generalization to Infinite Systems}\label{sec_infinity}
The important assumption of Convention \ref{conv_1}\,Êallowed a significant simplification of the results presented in this work. The theory can be generalized to cover infinite systems at the cost of topological complications. The generalized setting, however, allow one to fully appreciate the powerful mathematical set-up, based on linear algebra and functional analysis, on which the results of this work are developed. Here, we just sketch the main ideas.

One can define PNTS's on the category of compact Hausdorff topological spaces by replacing $\mathcal{P}$
with the functor $\mathcal{V}$ (mapping a space $X$ to the space $K(X)$ of its compact closed subsets, endowed with the Vietoris topology \cite{Kechris}) and $\mathcal{D}$ by $\mathcal{M}_{=1}$ (mapping a space $X$ to the space of probability measures on $X$, endowed with the $\textnormal{weak}^*\!$-topology \cite{Kechris}). Experiments on $X$ are now modeled by the space $C(X)$ of continuous functions $X\!\rightarrow\!\mathbb{R}$. This is a Riesz (Banach) space when endowed with the sup-norm \cite{LAX} and ordered pointwise. Both Theorem \ref{theorem_UE=convex} and the representation theorems of Section \ref{sec_representation_theorems} can be generalized to the new setting (see, e.g., \cite{Walley1991}). Similarly, counterparts of Theorem \ref{yosida_theorem} and Theorem \ref{metric_characterization} hold for general compact Hausdorff spaces (see, e.g., \cite{Luxemburgh,JVR1977} and \cite{LAX}). The results of Section \ref{real_valued_section} can then be proved \emph{mutatis mutandis}.

\bibliographystyle{abbrv}
\bibliography{biblio}

\end{document}